\documentclass[11pt]{article}
\usepackage[utf8]{inputenc}

\usepackage[bookmarksnumbered,colorlinks,citecolor=red,linkcolor=red,hyperindex,linktocpage=true]{hyperref}
\usepackage{authblk} 

\usepackage{amsthm}
\usepackage{amssymb}
\usepackage{amsmath}
\usepackage{mathtools}
\usepackage{bm,xcolor}
\usepackage{physics}
\usepackage{geometry}
\usepackage{qcircuit}

\usepackage{mathrsfs}

\usepackage{fullpage}
\newcommand{\mat}[1]{\begin{pmatrix}#1 \\ \end{pmatrix}}
\usepackage{graphicx,subfigure,epstopdf,float}

\usepackage[nameinlink,capitalise]{cleveref}

\usepackage{appendix}

\renewcommand{\i}{\ensuremath{\text{\normalfont I}}}

\def\R{\mathbb{R}}

\def\l{\lambda}
\newtheorem{theorem}{Theorem}[section]
\newtheorem{corollary}{Corollary}[theorem]
\newtheorem{lemma}[theorem]{Lemma}
\newtheorem{assumption}{Assumption}{\bf}{\it}
\newtheorem{remark}[theorem]{Remark}
\newtheorem{definition}[theorem]{Definition}

\DeclarePairedDelimiter\floor{\lfloor}{\rfloor}

\def\oL{\mathcal{L}}
\def\bxi{\boldsymbol{\xi}}
\def\bbf{\mathbf{f}}
\def\be{\mathbf{e}}

\def\to{\rightarrow}

\def\To{\Longrightarrow}

\def\oo{\infty}

\newcommand{\case}[2][lllllllllllllllllllllllllllllllllllll]{\left\{\begin{array}{#1}#2 \\ \end{array}\right.}

\title{Quantum Algorithms for Nonlinear Dynamics: Revisiting Carleman Linearization with No Dissipative Conditions
}
\author[1]{Hsuan-Cheng Wu\thanks{wu.hsuancheng@psu.edu}}
\author[2]{Jingyao Wang\thanks{jingyao-21@mails.tsinghua.edu.cn}}
\author[1]{Xiantao Li\thanks{Xiantao.Li@psu.edu}}

\affil[1]{Department of Mathematics, The Pennsylvania State University, University Park, Pennsylvania 16802, USA}
\affil[2]{Institute for Interdisciplinary Information Sciences, Tsinghua University, Beijing, China}

\begin{document}

\maketitle

\begin{abstract}
In this paper, we explore the embedding of nonlinear dynamical systems into linear ordinary differential equations (ODEs) via the Carleman linearization method.  Under {strong} dissipative conditions, numerous previous works have established rigorous error bounds and linear convergence for Carleman linearization, which have facilitated the identification of quantum advantages in simulating large-scale dynamical systems.
Our analysis extends these findings by exploring error bounds beyond the traditional dissipative condition, thereby broadening the scope of quantum computational benefits to a new class of dynamical regimes. This novel regime is defined by a resonance condition, and we prove how this resonance condition leads to a linear convergence with respect to the truncation level $N$ in Carleman linearization. We support our theoretical advancements with numerical experiments on a variety of models, including  the Burgers' equation, Fermi-Pasta-Ulam (FPU) chains, and the Korteweg-de Vries (KdV) equations, to validate our analysis and demonstrate the practical implications.

\end{abstract}

\section{Introduction}
 Nonlinear dynamical systems are ubiquitous in most scientific domains. The ability to simulate large-scale dynamical systems is crucial in predicting and controlling the dynamical behavior of the system.  This paper is concerned with the problem of designing quantum algorithms for solving dynamical systems, which might have the advantage when treating large-scale problems. Although many efficient quantum algorithms have been developed for linear ODE systems \cite{an2023linear,berry2017quantum, berry2022quantum,childs2019quantum,fang2023time,jin2022quantumshort,krovi2022improved}, algorithms for nonlinear dynamical systems with quantum advantage are much more difficult to construct. One regime where such an advantage has been identified is dissipative dynamics, where the real part of the eigenvalues of the Jacobian $F_1$ at an equilibrium is strictly negative, i.e., there exists a $\sigma>0$, 
 \begin{equation}\label{eq: dissip}
     \text{Re}(\lambda_j)\leq -\sigma<0, \quad \text{for all the eigenvalues.}
 \end{equation}
  In particular, for Carleman linearization \cite{carleman1932application},  global-in-time error bound has been proved \cite{amini2022carleman, forets2017explicit,liu2021efficient}  with a convergence rate
 \begin{equation}\label{eq: Rd}
     \mathsf{R}_d= \mathcal{O}\left(\frac{\mu \norm{F_2}}{\sigma}\right), 
 \end{equation}
 with $\norm{F_2}$ being the level of nonlinearity and $\mu$ representing the magnitude of the solution.  The analysis in \cite{liu2021efficient} also included external force. The quantum advantage then comes from the ability to efficiently simulate the resulting linear ODE system, a subject that has been relatively well-studied \cite{an2023linear,berry2017quantum, berry2022quantum,childs2019quantum,fang2023time,jin2024schr,jin2022quantumshort,jin2022quantum, krovi2022improved}. {In ODE theory, the concept of a dissipative property generally describes asymptotic stability near an equilibrium, but it can be mathematically weaker than the condition given in \eqref{eq: dissip}. Therefore, we will refer to \eqref{eq: dissip} as a strongly dissipative condition. }

Meanwhile, dynamical systems with some or all eigenvalues having zero real part are plentiful. The most recognizable example is Hamiltonian systems, which play a central role in many aspects of modern physics. The ability to treat such dynamical systems will certainly expand the scope of quantum computing algorithms. The lower bound from \cite{liu2021efficient} suggests that outside the dissipative regime, the Carleman linearization might not provide a convergent approximation, and quantum algorithms will generally not offer an advantage. On the other hand, there might be a subclass of non-dissipative dynamical systems for which a quantum advantage still exists. The purpose of this paper is to present a convergence analysis for Carleman linearization without the dissipative condition. One of our key contributions is the identification of another type of condition, called the \emph{no-resonance condition}, motivated by the theory of dynamical systems \cite{katok1995introduction}. Roughly speaking, this condition states that any one eigenvalue can not coincide with an integer combination of other eigenvalues. We quantify such condition by introducing a resonance parameter $\Delta$.  We prove the linear convergence of Carleman linearization with a convergence rate given by 
\begin{equation}\label{eq: Rr}
     \mathsf{R}_r= \mathcal{O}\left(\frac{\mu {\kappa_1(W)}\norm{F_2}_1}{\Delta}\right). 
 \end{equation}
 Therefore, for non-dissipative systems, $\Delta$  in \cref{eq: Rr} plays a similar role to $\abs{\sigma}$ in the convergence rate in the dissipative regime \cref{eq: Rd}. 

Resonance properties play a crucial role in the long-time behavior of dynamical systems, including phenomena such as bifurcation, chaos, and energy transport \cite{chow2012methods,katok1995introduction}. Important examples include Fermi-Pasta-Ulam (FPU) chains \cite{rink2001symmetry}, the nonlinear Schr\"odinger equation \cite{bernier2020rational}, and the Korteweg-de Vries (KdV) equation \cite{hiraoka2009normal}. Therefore, it should not come as a surprise that such properties play a role in an approximation scheme. On the other hand, due to their widespread occurrence in physics and engineering, developing quantum algorithms for these systems could significantly expand the scope of applications in quantum computing and create new paradigms for simulating complex dynamical systems.

The idea behind Carleman linearization \cite{carleman1932application} can be quickly illustrated using an ODE system with quadratic nonlinearity, 
\begin{equation}\label{eq:1.1}
    \frac{d}{dt}  \bm x = F_1 \bm x + F_2 \bm x \otimes \bm x, \, \bm x(0)=\bm x_{in}.
\end{equation}
Here $\bm x\in\R^n$, and $\bm x \otimes \bm x (=\bm x^{\otimes 2})\in \R^{n^2}$ denotes the tensor product of the two vectors.  Meanwhile, $F_1\in\R^{n\times n}$ is the Jacobian of the vector field at the equilibrium (here the origin), and $F_2: \R^{n^2}\to \R^n$ embodies the coefficient of the quadratic term.  Using Carleman linearization, one gets \cite{amini2022carleman, carleman1932application,tsiligiannis1989normal}
\begin{equation} \label{eq:1.2}
    \frac{d}{dt}\mat{\bm x \\ \bm x^{\otimes 2} \\ \vdots \\ \bm x^{\otimes (N-1)}  \\ \bm x^{\otimes N} }=\mathbf{A}\mat{\bm x \\ \bm x^{\otimes 2} \\ \vdots \\ \bm x^{\otimes (N-1)} \\ \bm x^{\otimes N} } + 
    \mat{0\\  0 \\  \vdots \\ 0 \\ A_{N,N+1} \bm x^{\otimes(N+1)}}. 
\end{equation}
Due to the quadratic nonlinearity, the matrix $\mathbf{A}$ is block upper triangular,
\begin{equation}\label{matrix-car}
\mathbf{A}=    \begin{pmatrix}
A_{1,1}  & A_{1,2} & 0& 0 & \cdots & 0 & 0 \\
0  & A_{2,2} & A_{2,3} &  0& \cdots & 0 & 0 \\
 0 & 0 & A_{3,3} &  A_{3,4}& \ddots & \vdots & \vdots \\
 0 & 0 & &A_{4,4}  & \ddots & \vdots & \vdots  \\
  \vdots  & \vdots  &  \vdots &\vdots  & \ddots & \vdots & \vdots  \\
 0 & 0 & 0 & 0 & \cdots & A_{N-1,N-1} & A_{N-1,N}  \\
 0 & 0 & 0 & 0& \cdots  & 0 & 
A_{N,N}\end{pmatrix}
\end{equation}
where the blocks are defined as
\begin{equation}\label{Ajj}
    \begin{aligned}
        A_{j,j}=&\sum_{i=1}^{j-1}I^{\otimes i}\otimes F_{1}\otimes I^{\otimes (j-1-i)} \in \mathbb{R}^{n^j \times n^j},\; 
        A_{j,j+1}=&\sum_{i=1}^{j-1}I^{\otimes i}\otimes F_{2}\otimes I^{\otimes (j-1-i)} \in \mathbb{R}^{n^j \times n^{j+1}}.
    \end{aligned}
\end{equation}
$I$ refers to the identity matrix acting on $\mathbb{R}^n.$ General nonlinear ODE systems where the vector field is a polynomial can be reduced to the quadratic form in \cref{eq:1.2}, as illustrated in \cite{forets2017explicit}.

To obtain a finite ODE system, we truncate the system by neglecting the last term in \cref{eq:1.2}, yielding
\begin{equation}\label{y-car}
     \frac{d}{dt} \begin{pmatrix}
          \bm y_1 \\
          \bm y_2 \\ 
          \vdots \\
          \bm y_N
     \end{pmatrix} = \mathbf{A} \begin{pmatrix}
          \bm y_1 \\
          \bm y_2 \\ 
          \vdots \\
          \bm y_N
     \end{pmatrix},
\end{equation}
where $\bm y_j\in\R^{n^j} $ takes the position of $\bm x^{\otimes j}$; $\bm y_j(0)=\bm x_{in}^{\otimes j}$. In particular, we will regard the first block entry $\bm y_1(t)$ as an approximation to $\bm x(t).$

The analysis \cite{amini2022carleman, forets2017explicit,liu2021efficient} of the error due to the finite truncation has been mostly based on the assumption that the eigenvalues of $F_1$ have negative real parts \eqref{eq: dissip}, and when such condition does not hold, the linear convergence was only proved for short time, which significantly limits the applicability of the algorithms. Here we propose a different type of condition, which can go well beyond the dissipative regime \eqref{eq: dissip}.   

\begin{assumption}\label{assump1}
    Assume that $F_1$ is diagonalizable, {
    \begin{equation}\label{eq: F1-eig}
    F_1= W \Lambda W^{-1},
    \end{equation} }
    and it has eigenvalues $\{\l_1,\cdots,\l_n\}$ with non-positive real parts: $\text{Re}(\lambda_j)\leq 0.$ In addition, for any $i\in[n]$, the following holds
    \begin{equation}\label{cond-res}
        \l_i \ne \sum_{j=1}^n m_j\l_j,  \quad \forall m_j\in \mathbb{Z} \; \text{and} \; m_j \ge0 \text{ s.t. } \sum_{j=1}^n m_j\ge2. 
    \end{equation}
   More precisely, we assume that $\Delta>0$, defined as follows, 
    \begin{equation}\label{Delta}
        \Delta := \inf_{k\in[n]}\inf_{\overset{m_j\ge0} {\sum_{j=1}^n m_j\ge2} } \left|\l_k-\sum_{i=1}^n m_{j}\l_j \right|. 
    \end{equation}
\end{assumption}
{
We denote $\bm e_1,\cdots,\bm e_n$ and $\bm f_1^T,\cdots,\bm f_n^T$ are the right and left eigenvectors of $F_1$, correspondingly, which will form the columns/rows of the matrix $W$ and $W^{-1},$ respectively.  $\kappa(W)=\norm{W}\norm{W^{-1}}$ is the condition number of $W$. The subscript of $\kappa$ refers to which norm is used. {For the first part of the analysis, $\norm{\bullet}$ referss to any vector-induced norm. We will specify the vector norm when necessary, e.g.,  $\norm{\bullet}_1$ for the $\ell_1$ norm.}
}
In the theory of dynamical systems \cite{katok1995introduction}, the first part of the assumption comes from stability consideration, but we do not rule out non-hyperbolic equilibrium points.  Meanwhile, the condition in \cref{cond-res} is known as the no-resonance condition. When $ \l_i = \sum_{j=1}^n m_j\l_j$ for some $m_j\ge0$ and $\sum_{j=1}^n m_j=m\ge2$, we call the corresponding eigen modes a resonance of order $m$. One implication of such conditions is  that a nonlinear dynamical system can be reduced to a linear ODE system, or a nonlinear ODE system with much fewer terms, via a polynomial transformation. The corresponding reduced form is known as the normal form. The connection of such reduction with the Carleman linearization has been noted in \cite{tsiligiannis1989normal}.

To quantify such resonance property and how this condition determines the error in the Carleman linearization in \cref{{y-car}},  we refer to $\Delta$ as the \emph{resonance parameter}. Our main result can be summarized as follows,

\begin{theorem}\label{thm:main}
     Assume that the solution of the nonlinear ODE in \cref{eq:1.1} satisfies a uniform-in-time bound in $t \in [0,T]$,
    \begin{equation}
        \norm{\bm x(t)}_1 \leq \mu,
    \end{equation}
and $F_1$ fulfills \cref{assump1}. Then the error in the Carleman linearization can be bounded by {
\begin{equation}
    \norm{\bm x(T) - \bm y_1(T) }_1 \leq N CT \mathsf{R}_r^{N-1},
\end{equation}
where $C$ and $\mathsf{R}_r$ are defined by,
\begin{equation}
    C\coloneqq\kappa_1(W)\norm{F_2}_1\mu^2,\quad \mathsf{R}_r:= \frac{ 4e\mu \kappa_1(W)\norm{F_2}_1}{\Delta}.
\end{equation}
}
\end{theorem}

Notice that the convergence rate $\mathsf{R}_r$ is independent of $T$. This result suggests the choice of 
\begin{equation}
    N= \mathcal{O} \left( \frac{\log \frac{T}{\epsilon}  }{\log\frac{1}{\mathsf{R}_r}} \right),
\end{equation}
 in order to keep the truncation error to within a tolerance $\epsilon>0$.

\paragraph{The complexity of the quantum algorithm}

One approach to solve the ODEs from the $N$-th truncation \cref{y-car} of the Carleman system is by using a time discretization in conjunction with quantum linear solvers, e.g., \cite{berry2017quantum}.  Due to the fact that the {quantum linear system algorithms (QLSA) incur a near-optimal scaling \cite{an2022quantum, costa2022optimal,jennings2023efficient,lin2020optimal}}, this approach can offer exponential speedup for large-dimensional dynamical systems.  
\cite{liu2021efficient} chose the Euler's method as an example which has first order and the final query complexity to $A$ of the quantum algorithm has a scaling (aside from logarithmic factors) of $\widetilde{\mathcal{O}}\left(\frac{q T^2}{\epsilon}\right)$, where $q= \frac{\norm{\bm y(0)}}{\norm{\bm y(T)}}$. \cite{krovi2022improved} combined high-order ODE solvers and obtained a quantum algorithm with complexity $\widetilde{\mathcal{O}}\big(q T \norm{A} \alpha\big)$, where $q= \frac{\max_{t\in[0,T]}\norm{\bm y(t)}}{\norm{\bm y(T)}}$ and $\alpha= \max_{t\in[0,T]} \norm{\exp(-tA) }$. See \cite{fang2023time} for a detailed discussion about various quantum ODE solvers. {On the other hand, there are non-QLSA based quantum methods for solving linear differential equations that can also be useful, including time-marching strategy \cite{fang2023time} and linear combination of Hamiltonian simultations \cite{an2023quantum,an2023linear}.}

\paragraph{Related works} 
Quantum algorithms based on Carleman linearization have been extended to fluid mechanics problems \cite{itani2024quantum,li2023potential}, chemical kinetics \cite{akiba2023carleman} and reaction-diffusion equations \cite{liu2023efficient}. 
On the theoretical side, Carleman linearization can be extended to dynamical systems where the vector field admits a polynomial approximation \cite{amini2022carleman, surana2022efficient}.
The work \cite{amini2022carleman} shares similar objectives as this paper. Their main assumption {\bf 3.1} is essentially the dissipative condition in \eqref{eq: dissip}. They also established an error bound without this assumption (Theorem 3.1), but the linear convergence with respect to $N$ only holds for short time, which is similar to the observation in \cite{forets2017explicit} and \cite{liu2021efficient}.    

 There have been numerous attempts to extend the Carleman linearization to non-dissipative systems \cite{akiba2023carleman,succi2024three}. However, the explicit error bounds have not been provided in these works. \cite{lewis2023limitations} identified certain chaotic dynamical systems where quantum advantage does not exist, which extends the lower bound in \cite{liu2021efficient}. 
Meanwhile, there are other linear embedding schemes \cite{engel2021linear,giannakis2022embedding}. It is likely that the error bounds for these embedding schemes also require spectral properties of the dynamical system. 

\cite{joseph2020koopman} proposed to map nonlinear ODEs to a linear PDE that governs the time evolution of the probability density function. \cite{jin2022quantum-ob} follows a similar track, but focused on the computation of observables. In this approach, the overall complexity has at least a polynomial dependence on the dimension $n$ \cite{jin2023time}.

Variational quantum algorithms have also been proposed for solving nonlinear differential equations \cite{kyriienko2021solving, lubasch2020variational,  pool2024nonlinear}. While these algorithms are more feasible for near-term devices, theoretical analyses of approximation errors and quantum advantages have remained largely unexplored.

The remainder of this paper is structured as follows:  \cref{sec: err} details our analysis of the truncation errors resulting from Carleman linearization, including an in-depth investigation of the eigenvectors associated with matrix $\mathbf{A}$. Subsequently, in \cref{sec: num}, we present a series of numerical experiments conducted on various nonlinear dynamical systems.

\section{Error Analysis of Carleman Linearization}\label{sec: err}
{ Throughout the paper, unless otherwise specified, $\norm{\bm v}$ denotes a vector norm of $\bm v$, and $\norm{A}$ denotes the corresponding induced matrix norm.} 

Recall that $\bf{A}$ is the matrix from the Carleman linearization \cref{matrix-car}, and $\bm y(t)$ represents the corresponding solution: 
\begin{equation}\label{ode-l}
 \frac{d}{dt}   {\bm y}(t) = \mathbf{A} \bm y(t).
\end{equation}
To prepare for the error analysis, we follow \cite{forets2017explicit,liu2021efficient} and define
\begin{equation}
    \bm \eta = (\bm \eta_1, \bm \eta_2, \cdots, \bm \eta_N ), \quad \bm \eta_j(t):= \bm x^{\otimes j}(t) - \bm y_j(t).
\end{equation}

We notice that by choosing proper initial conditions for $\bm y(0)$, we get $ \bm \eta(0)=0$. In addition, it satisfies a non-homogeneous ODE with the same matrix $\mathbf{A}$,
\begin{equation}
     \frac{d}{dt}   \bm \eta(t) = \bf{A}  \bm \eta(t)  +  \mat{0\\    \vdots \\ 0 \\ A_{N,N+1} \bm x^{\otimes(N+1)}}. 
\end{equation}
Therefore, the final error that needs to be estimated is $\bm \eta_1(t).$

\subsection{A direct estimate using nested integrals}

We begin with some observations on the spectral properties of $\bf{A}.$
\begin{lemma}\label{lem:1}
    Assume that $F_1$ has eigenvalues $\{\lambda_1,\cdots,\lambda_n\}$, then the eigenvalues  $\bf{A}$ are given by
    $$\left\{\sum_{k=1}^n m_{k}\lambda_k | m_k \in \mathbb{N}\cup \{0\} \; \mathrm{and} \; \sum_{k=1}^n m_{k}\leq N \right\}. $$
    In addition, if $F_1$ is diagonalizable with the real part of the eigenvalues being non-positive, then each block $A_{j,j}$ is a stable matrix, in that its eigenvalues have non-positive real parts, and those eigenvalues with zero real parts have the geometric multiplicity being 1.
\end{lemma}
\begin{proof}
    Note that $\bf{A}$ is a block upper triangular matrix; therefore, the eigenvalues of $\bf{A}$ are the union of the eigenvalues of the diagonal blocks $A_{j,j}$. Further, notice in \cref{Ajj} that
     eigenvalues of Kronecker product matrices are the product of eigenvalues of the matrices and we immediately have the eigenvalues $ \sum_{k=1}^j\l_{i_k},$ where $i_k \in [n].$  Let $\bm v_j, j=1,2,\cdots, n$ be the eigenvectors of $F_1$. Then the corresponding eigenvector of this eigenvalue is $\bm v_{i_1} \otimes \bm v_{i_2} \otimes\bm v_{i_3} \otimes \cdots \otimes \bm v_{i_j}.$ As a result, all eigenvalues of $A_{j,j}$ have non-positive real part with geometric multiplicity being 1, which shows the stability of the matrix $A_{j,j}$. 
\end{proof}

\begin{theorem}
Assume that $\norm{\bm x(t)} \leq \mu$ for all $t\geq 0,$ and $F_1$ is diagonalizable with the real part of the eigenvalues being non-positive. Then the following error bound holds,
    \begin{equation}
    \norm{\bm \eta_1(T)} \leq \mu \big(\mu T \norm{F_2}\big)^N. 
\end{equation}
\end{theorem}

\begin{proof}
    We start with the first equation in \cref{ode-l} and apply the variation of constant formula,
\begin{align*}
     \bm \eta_1(t)&=\int_0^t e^{(t-t_1)A_{1,1}} A_{1,2}   \bm \eta_2(t_1 ) dt_1 
        = \int_0^t \int_0^{t_1} e^{(t-t_1)A_{1,1}} A_{1,2} e^{(t-t_1)A_{2,2}} A_{2,3}    \bm \eta_2(t_1 ) dt_2 dt_1.   
\end{align*}
The first matrix exponentials only involve $F_1$. From the stability result in \cref{lem:1}, the matrix exponential is uniformly bounded in time.  For the second matrix exponential, we apply \cref{lem:1} again and still obtain the uniform boundedness in time.   Therefore, by repeating these steps, we have
\begin{equation}\label{err-eta1}
\begin{aligned}
       \norm{\bm \eta_1(t)} \leq &  \int_0^t \int_0^{t_1} \cdots \int_0^{t_{N-1}}  \norm{ A_{1,2}}   \norm{ A_{2,3}} \cdots  \norm{ A_{N-1,N}} \norm{\bm \eta_N(t_{N-1}) } dt_{N-1} \cdots dt_2 dt_1\\
       \leq  &  \int_0^t  \cdots   \int_0^{t_{N}} \norm{ A_{1,2}}    \cdots  \norm{ A_{N-1,N}} \norm{\bm \eta_N(t_{N-1}) } N \norm{F_2} \norm{\bm x(t_N)}^{N+1} dt_N  \cdots dt_2 dt_1 \\
       \leq & \mu (t \norm{F_2} \mu)^N. 
\end{aligned}
\end{equation}
Here the nested integral is reduced to $\frac{t^N}{N !}$. 

\end{proof}

This result has been established in  \cite[Eq. (62)]{liu2021efficient} (and a similar result in \cite[Theorem 4.3]{forets2017explicit}). We included the proof here to provide a comparison with the proof for the non-dissipative case in the next section, where the convergence rate is independent of $T$.  In particular, this estimate suggests without dissipative structures in $F_1$,  a linear convergence can only be guaranteed for a short time. In the next section, we will provide analysis using the no-resonance condition \cref{cond-res}, which can hold even without the dissipative condition.

\subsection{Improved estimates using the no-resonance condition}

We first give a high-level description of the analysis. Assume that $F_1$ has been mapped to a diagonal matrix by a similar transformation and that  $\bf {A}$ admits a  diagonalization $VD V^{-1}=\bf {A}$ (see \cref{lem:1}), then we can write the solution of the error equation \cref{ode-l} as
\begin{equation}\label{eta-zeta}
     \bm \eta(t) =  \int_0^t V e^{\tau D} \bm \zeta(t-\tau) d\tau, \quad \bm \zeta(t-\tau):=  V^{-1}  \mat{0\\  \vdots \\ 0 \\ A_{N,N+1} \bm x^{\otimes(N+1)}(t-\tau)}.
\end{equation}

In light of \cref{lem:1}, we have $\norm{e^{\tau D}} \leq 1$ due to stability, and as a result, we have a time-independent bound on the matrix exponential. It remains to estimate the bound of the rest of the terms. In particular, since $V$ and $V^{-1}$ are time-independent, we will not encounter a nested integral and a $t^N$ term in the error bound. Furthermore,  we notice that $V$ is block upper triangular and the diagonals can be set to identity matrices. 
As a result, $\forall t\geq 0,$ we have
\begin{equation}\label{eq:zeta2eta}
 \begin{aligned}
         \norm{\bm \eta(t)_1}& \leq \int_0^t \norm{\bm \zeta_1} +  \norm{V_{12}}\norm{\bm \zeta_2} + \cdots +   \norm{V_{1N}}\norm{\bm \zeta_N} d\tau.
 \end{aligned}
\end{equation}
Therefore the problem depends critically on the structure of the eigenvectors in $V$.

Toward this end, we observe that $\bf{A}$ is block bi-diagonal,  and we proceed to analyze the structure of the eigenvectors. Since the eigenvalues of $\bf{A}$ coincide with the union of those eigenvalues of $A_{j,j}$, the problem can be reduced to,
\begin{equation}\label{eq: evector}
\begin{bmatrix}
 \ddots & \vdots & \vdots& \vdots& \vdots& \vdots& \vdots \\
  \cdots  &0 & A_{j-2, j-2} &   A_{j-2, j-1} & 0 & 0 & \cdots \\
  \cdots  & 0 &0  & A_{j-1, j-1} & A_{j-1, j} & 0 & \cdots \\
  \cdots  &0 &0 & 0 &  A_{j, j} &  A_{j,j+1}  & \cdots \\
    \vdots & \vdots & \vdots& \vdots& \vdots& \vdots& \ddots
\end{bmatrix}
\begin{bmatrix}
    \bm w_1 \\
    \vdots \\
    \bm w_{j-2}\\
    \bm w_{j-1} \\
     \bm w_j\\
     0 \\
     \vdots \\
\end{bmatrix} 
=
\lambda 
\begin{bmatrix}
    \bm w_1 \\
    \vdots \\
    \bm w_{j-2}\\
    \bm w_{j-1} \\
     \bm w_j\\
     0 \\
     \vdots \\
\end{bmatrix}.     
\end{equation}

Since we have assumed that $F_1$ is diagonal,  the eigenvectors, denoted here by $\bm e_i \in \mathbb{R}^n$, form a complete basis.  
From the $j$-th component, we have $ A_{j, j} \bm w_j = \lambda \bm w_j. $ Therefore,
\begin{equation}\label{A_jj-eig}
    \lambda = \lambda_{i_1} + \lambda_{i_2} + \cdots + \lambda_{i_j}, \quad \bm w_j = \bm e_{i_1} \otimes \bm e_{i_2} \otimes \cdots \otimes \bm e_{i_j},
\end{equation}
for some tuple $i_1, i_2, \cdots, i_j \in [n].$ 
A back substitution yields 
\begin{equation}\label{w-j-1}
    \begin{aligned}
    \bm w_{j-1} = & ( \lambda I - A_{j-1,j-1}  )^{-1} A_{j-1,j} \bm w_j \\
         = & ( \lambda I - A_{j-1,j-1}  )^{-1}  \Big( (F_2\bm  e_{i_1} \otimes \bm e_{i_2} ) \otimes \bm e_{i_3} \otimes \cdots  \otimes \bm e_{i_j} \\
         &\qquad  \quad   + \cdots + \bm  e_{i_1} \otimes \bm e_{i_2} 
         \otimes \cdots \otimes \bm e_{i_{j-2}} (F_2 \bm e_{i_{j-1}} \otimes \bm e_{i_j} ) \Big).
\end{aligned}
\end{equation}

By spectral decomposition, we have,
{
\begin{equation*}
    A_{j-1,j-1} = \sum_{i_1',i_2',\cdots, i_{j-1}'} (\l_{i_1'} + \l_{i_2'} + \cdots +  \l_{i_{j-1}'}) 
\left( \bm e_{i_1'} \otimes \bm e_{i_2'}  \otimes \cdots  \otimes \bm e_{i_{j-1}'} \right) \left( \bm f_{i_1'} \otimes \bm f_{i_2'}  \otimes \cdots  \otimes \bm f_{i_{j-1}'} \right)^T.
\end{equation*} 
}
The summation here is over $i_1',i_2',\cdots, i_{j-1}' \in[n].$

First, for simplicity, assume that  $A_{j,j}$ and $A_{j-1,j-1}$ do not share eigenvalues, i.e., $\lambda I - A_{j-1,j-1}$ is invertible. Then one has
{
\begin{equation*}
\begin{aligned}
         ( \lambda I - A_{j-1,j-1}  )^{-1} = \sum_{i_1',i_2',\cdots, i_{j-1}'} \frac{1}{\lambda - (\l_{i_1'} +  \cdots +  \l_{i_{j-1}'}) } 
 \times \left( \bm e_{i_1'}  \otimes \cdots  \otimes \bm e_{i_{j-1}'} \right) \left(\bm f_{i_1'} \otimes \cdots  \otimes \bm f_{i_{j-1}'} \right)^T.
\end{aligned}
\end{equation*}
Recall that the vectors $\bm f_i^T \in \mathbb{C}^{1\times n}$ are the left eigenvectors of $F_1$. 
 }

To simplify notations, let us define,
\begin{equation}\label{eq: xi-1}
    \bm \xi_k = F_2 \bm  e_{i_k} \otimes \bm  e_{i_{k+1}}, 
\end{equation}
and examine, {
\begin{equation}
       ( \lambda I - A_{j-1,j-1}  )^{-1}  \bm  \xi_1 \otimes  \bm e_{i_3} \otimes \cdots   \bm \otimes \bm  e_{i_j} 
      =  \sum_{i_1'} \frac{1}{\lambda_{i_1} + \lambda_{i_2} - \l_{i_1'}  }  \bm e_{i_1'}  \big(\bm f_{i_1'}^T\bm  \xi_1\big) \otimes    \bm e_{i_3} \otimes \cdots  \otimes \bm  e_{i_j}.
\end{equation}
}
{
Notice that $\bm f_i^T\bm e_j = \delta_{i,j}$ by the orthogonality between left and right eigenvectors. We had to set $i_2'=i_3,\cdots, i_{j-1}'=i_j$ to get a nonzero inner product between  $\left(\bm f_{i_1'} \otimes \cdots  \otimes \bm f_{i_{j-1}'} \right)^T$ and $\bm\xi_1\otimes\bm e_{i_3}\otimes\cdots\otimes\bm e_{i_j}$.

}

Finally, due to the tensor-product form, we can limit the summation to $ \bm e_{i_1'} \bm  f_{i_1'}^\dag$. Therefore we can collect these terms and define,
{
\begin{equation}\label{G12}
    G_{12}= \sum_{i_1'} \frac{1}{\lambda_{i_1} + \lambda_{i_2} - \l_{i_1'}  }  \bm  e_{i_1'}  \bm  f_{i_1'}^T
    = \Big[ (\lambda_{i_1} + \lambda_{i_2})I -F_1 \Big]^{-1}.
\end{equation}
}

By repeating this argument, we find that,
\begin{equation}\label{eq: w_j-1}
\begin{aligned}
      \bm w_{j-1}=&\; G_{12}  \bm \xi_1 \otimes \bm  e_{i_3} \otimes  \bm e_{i_4} \otimes  \cdots \otimes \bm  e_{i_j}  
      +   \;  \bm    e_{i_1} \otimes G_{23} \bm  \xi_2 \otimes \bm  e_{i_4} \otimes  \cdots \otimes  \bm e_{i_j}   \\
      + & \; \bm e_{i_1} \otimes 
    \bm  e_{i_2} \otimes  G_{34}  \bm \xi_3 \otimes  \bm e_{i_5} \otimes  \cdots \otimes  \bm e_{i_j} 
    + \;\cdots 
    +   \; \bm e_{i_1} \otimes 
    \bm  e_{i_2} \otimes \cdots \otimes   \bm e_{i_{j-2}}  \otimes G_{j-1,j} \bm  \xi_{j-1}.
\end{aligned}
\end{equation}

The first subtlety in this analysis comes from the careful organization of the many terms that emerge, e.g., from applying the $j-2$ terms in $A_{j-2,j-1}$ when computing $A_{j-2,j-1} \bm w_{j-1}.$ Intuitively, we apply $F_2$ to consecutive vectors in \cref{eq: w_j-1}. When these two consecutive vectors are the eigenvectors, e.g., $\bm e_{i_k} \otimes \bm e_{i_{k+1}},$ we obtain $\bm \xi_k.$ Meanwhile, one of the vectors can also be $\bm \xi$. In this case, we will use a pairing scheme, illustrated as follows,
\begin{equation}
    \Qcircuit @C=1em @R=.7em {
 & \gate{\bm e_{i_1}}&\gate{\bm e_{i_2}} &  \cdots &
 &\gate{G_{k,k+1}\bm \xi_k }&\gate{\bm e_{i_{k+2}}} &\qw&\gate{\bm e_{i_{k+3}}}& \cdots&   &\gate{\bm e_{i_j}} &  \\
  & \gate{\bm e_{i_1}}&\gate{\bm e_{i_2}} &  \cdots &
 &\gate{\bm e_{i_{k}} }&\gate{G_{k+1,k+2}\bm \xi_{k+1} } & \qw &\gate{\bm e_{i_{k+3}}}& \cdots&   &\gate{\bm e_{i_j}} &  \\
& && & F_2\Big( & && \Big)\quad  & & & &\\
}\\
\end{equation}

Motivated by this observation, we generalize the definition of $\bm \xi $ in \cref{eq: xi-1} as follows: 
\begin{definition}
    Let $\bxi_m^{(k)} \in \mathbb{R}^n$ be defined recursively as follows, $ \bxi_m^{(0)} = \be_{i_m},$ and 
\begin{equation}
    \begin{aligned}
    \bxi_m^{(k)} &= \sum_{a=1}^{k-1} F_2\left(G_{p_m}\bxi_m^{(a)}\otimes G_{p_{m+a+1}}\bxi_{m+a+1}^{(k-a)}\right) \quad \text{ for} \; \text{all} \;  k\geq1.
\end{aligned}
\end{equation}
\end{definition}
Here $G_{p_m}$ is a generalization of \cref{G12} and it will be defined later.
With these definitions and lengthy calculations, we have,
\begin{equation}\label{Aj-2j-1wj-1}
    \begin{aligned}
    A_{j-2,j-1} \bm w_{j-1} &= \sum_{k=1}^{j-3} \sum_{k'=k+2}^{j-2}
     \Big[ \bm e_{i_1} \otimes \bm  e_{i_{k-2}} G_{k,k+1}  \bm \xi_k^{(1)} \otimes \cdots \otimes  \bm e_{i_{k'-1}} \otimes \bm  \xi_{k'}^{(1)} \otimes  \bm  e_{i_{k'+2}} \otimes \cdots \bm  e_{i_j} \\
     &  + \bm  e_{i_1} \otimes  \bm e_{i_{k-2}}  \bm  \xi_k^{(1)} \otimes \cdots \otimes  \bm e_{i_{k'-1}} \otimes G_{k',k'+1}  \bm \xi_{k'}^{(1)} \otimes   \bm e_{i_{k'+2}} \otimes \cdots  \bm e_{i_j}
     \Big] \\
    & +  \bm \xi_1^{(2)} \otimes  \bm  e_{i_4} \otimes  \otimes  \bm  e_{i_5} \cdots \otimes  \bm e_{i_j}  
     +  \bm e_{i_1} \otimes  \bm \xi_2^{(2)} \otimes   \bm e_{i_5} \otimes  \cdots \otimes  \bm e_{i_j}  \\
    & + \cdots 
     +  \bm e_{i_1} \otimes 
    \bm  e_{i_2} \otimes \cdots \otimes  \bm e_{i_{j-3}} \otimes  \bm \xi_{j-2}^{(2)}. 
\end{aligned}
\end{equation}

The other technical subtlety in the analysis comes from the inversion of $\lambda I - A_{j-2,j-2} $ as well as the inversion of such matrices in the subsequent steps. In the previous step, the existence of the inverse is guaranteed by the no-resonance condition \eqref{cond-res}, which can be seen in the definition \cref{G12}. However, the matrix inverse in later steps will require the Fredholm alternative. To elaborate on this point, we first notice that, $\left(\lambda I - A_{j-2,j-2} \right) \bm w_{j-2} =  A_{j-2,j-1} \bm w_{j-1}.$
The matrix on the left hand side has eigenvalues $\lambda - \lambda_{i_1'} - \lambda_{i_2'} -\cdots  - \lambda_{i_{j-2}'}, $ which might become zero even under the no-resonance condition \eqref{cond-res}. Let us first examine a term from \cref{Aj-2j-1wj-1} that involves $\bm \xi^{(2)},$ e.g., we can study the following linear system,
\(
 \left(\lambda I - A_{j-2,j-2} \right) \bm w = \bm b, \quad \bm b:=\bm \xi_1^{(2)} \otimes  \bm  e_{i_4} \otimes  \bm  e_{i_5} \cdots \otimes  \bm e_{i_j}.
\)

{
The eigenvalue of $\lambda I - A_{j-2,j-2}$ is $\lambda - \lambda_{i_1'} - \lambda_{i_2'} - \cdots  -\lambda_{i_{j-2}'} $ with the corresponding left eigenvector $\bm z:=\left(\bm f_{i_1'} \otimes \bm f_{i_2'} \otimes  \cdots \otimes \bm f_{i_{j-2}'}\right)^T.$ In fact, $\bm z \cdot \bm b\ne0$ only when $i_2'=i_4,$ $i_3'=i_5, \cdots $ and $i_{j-2}'=i_j.$ In this case, the eigenvalue becomes $\lambda_{i_1}+\lambda_{i_2}+\lambda_{i_3}-\lambda_{i_1'} $, which can not be zero due to the no-resonance condition in \cref{cond-res}, and leads to a contradiction.

}

Let use examine another type of terms in \cref{Aj-2j-1wj-1}, by considering,
\[
 \left(\lambda I - A_{j-2,j-2} \right) \bm w = \bm b, \quad \bm b:=G_{12} 
 \bm \xi_1^{(1)}\otimes \bm \xi_3^{(1)}\otimes \bm e_{i_5} \otimes \cdots\otimes \bm  e_{i_j} +  
 \bm \xi_1^{(1)}\otimes G_{34}  \bm \xi_3^{(1)}\otimes \bm e_{i_5} \otimes \cdots\otimes \bm  e_{i_j} .
\]
In this case, in the eigenvectors, we must choose $i_3'=i_5, \cdots $ and $i_{j-2}'=i_j.$
Therefore, the eigenvalue is reduced to $(\l_{i_1}+\l_{i_2}+\l_{i_3}+\l_{i_4})-(\l_{i_1'}+\l_{i_2'})$, which might be zero even under the condition \eqref{cond-res}.  It is enough to examine the first two dimensions. {We notice that,
\[ 
\bm f_{i_1'}^T G_{12} =    \frac{1}{\l_{i_1}+\l_{i_{2}}-\l_{i_1'}} \bm f_{i_1'}^T, \; \;
\bm f_{i_2'}^T G_{34} =  \frac{1}{\l_{i_3}+\l_{i_4}-\l_{i_2'}}\bm f_{i_2'}^T.
\]
Therefore,  in this case, the inner product with the right hand side becomes
\[ 
\bm z\cdot  \bm b= \frac{(\l_{i_1}+\l_{i_2}+\l_{i_3}+\l_{i_4})-(\l_{i_1'}+\l_{i_2'})}{(\l_{i_1}+\l_{i_{2}}-\l_{i_1'} ) (\l_{i_3}+\l_{i_4}-\l_{i_2'}) }   \bm f_{i_1'}\cdot  \bm \xi_1^{(1)} \bm f_{i_2'}\cdot  \bm \xi_3^{(1)} =0, 
\]
provided that $\bm z$ is in the null space.} Therefore, the Fredholm alternative still applies. By applying the same argument to all the terms in \cref{Aj-2j-1wj-1}, we find a solution for $\bm w_{j-2}$,
\begin{equation}\label{w_j-2 Fredholm}
    \begin{aligned}
         \bm w_{j-2} =& \sum_{k_1=1}^{j-3} \sum_{k_2=k_1+2}^{j-2}
      \bm e_{i_1} \otimes \bm  e_{i_{k_1-1}} G_{k_1,k_1+1}  \bm \xi_{k_1}^{(1)} \otimes \cdots \otimes  \bm e_{i_{k_2-1}} \otimes G_{k_2,k_2+1}  \bm  \xi_{k_2}^{(1)} \otimes  \bm  e_{i_{k_2+2}} \otimes \cdots \bm  e_{i_j}  \\
    & +  G_{123} \bm \xi_1^{(2)} \otimes  \bm  e_{i_4} \otimes  \otimes  \bm  e_{i_5} \cdots \otimes  \bm e_{i_j}  
     +  \bm e_{i_1} \otimes  G_{234} \bm \xi_2^{(2)} \otimes   \bm e_{i_5} \otimes  \cdots \otimes  \bm e_{i_j}  \\
    & + \cdots +  \bm e_{i_1} \otimes 
    \bm  e_{i_2} \otimes \cdots \otimes  \bm e_{i_{j-3}} \otimes  G_{j-2,j-1,j} \bm \xi_{j-2}^{(2)}.
      \end{aligned}
\end{equation}

Here $G_{123}$ and similar matrices are defined as follows, 

\begin{definition}
Let $p_r(m_i,k_i)= m_i:(m_i+k_i):= \{m_i,m_i+1,\cdots,m_i+k_i\}$ for $i\in[N]$. 
Let $G\in \mathbb{C}^{n\times n}$ be defined as follows,
\begin{equation}\label{Gmat}
     G_{i_k:i_k'} = \left((\l_{i_k}+\cdots+\l_{i_k'})I-F_1\right)^{-1}.
\end{equation}
\end{definition}

From the no-resonance condition \eqref{cond-res},  we immediately have,
\begin{lemma}\label{lem: Gnorm}
The matrices $G$'s in \cref{Gmat} are well defined, and they are uniformly bounded, 
    \begin{equation}
        \norm{G} \leq{\frac{\kappa(W)}{\Delta}}\;,
    \end{equation}
{ where $\kappa(W)= \norm{W} \norm{W^{-1}}$ and the above inequality holds  for any induced norm.   }
\end{lemma}

\begin{lemma}\label{Lem: Main lemma} Let $\bm w$ be the eigenvector from \cref{eq: evector} with $j=N$. Under the no-resonance condition \eqref{cond-res}, the components of $\bm w$ can be written as,
    \begin{equation}\label{w_N-k}
        \bm w_{N-k}= \sum_{k_1+\cdots+k_N=k}\sum_{m_1,\cdots, m_N}\mathcal{L}(G_{p_1}\bxi_{m_1}^{(k_1)},\cdots,G_{p_l}\bxi_{m_N}^{(k_N)})
    \end{equation}
for $k=1,\cdots,N-1$. Here the operator $\mathcal{L}(G_{p_1}\bxi_{m_1}^{(k_1)},\cdots,G_{p_l}\bxi_{m_N}^{(k_N)})$ is defined by 
starting with the string $\be_{i_1}\otimes\be_{i_2}\otimes\cdots\otimes\be_{i_N},$ then replace $\be_{i_{m_1}} \otimes \cdots \otimes \be_{i_{m_1+k_1-1}}$ with $G_{p_1}\bxi_{m_1}^{(k_1)}$, and replace $\be_{i_{m_2}} \otimes \cdots \otimes \be_{i_{m_2+k_2-1}}$ with $G_{p_2}\bxi_{m_2}^{(k_2)}$, $\cdots$, and so on.
\end{lemma}
The proof, due to the lengthy calculations, is deferred to \cref{a-proof}.
Due to the upper triangular structure of $\mathbf{A}$, this formula can be extended to other eigenvectors,
\begin{theorem} \label{eigenvector thm}
    For any fixed $j=1,\cdots,N$, we have 
     \begin{equation}\label{w_j-k}
        \bm w_{j-k}= \sum_{k_1+\cdots+k_j=k}\sum_{m_1,\cdots, m_j}\mathcal{L}(G_{p_1}\bxi_{m_1}^{(k_1)},\cdots,G_{p_j}\bxi_{m_j}^{(k_j)}) .
    \end{equation}
\end{theorem}
\begin{proof}
    The proof directly follows from \cref{Lem: Main lemma}.
\end{proof}

We now move on to estimate these components of the eigenvectors. 
\begin{lemma}\label{lem: xik}
{ Following from the result in \cref{lem: Gnorm} that $\norm{G}\leq \frac{\kappa(W)}{\Delta}$, and the eigenvectors are normalized $\norm{\bm e_i}=1 $,}
   the vectors $\bxi^{(k)}$ have the following bound,
    \begin{equation}
        \norm{\bxi^{(k)}}\leq \frac{{\kappa(W)^{k-1}}\norm{F_2}^k}{\Delta^{k-1}} \frac{(2k)!}{(k+1)!k!}.
    \end{equation}
\end{lemma}
\begin{proof}
    Using the recursion relation \eqref{eq:xi^(c+1)}, together with the Catalan sequence, we deduce this bound.  Specifically, we can assume that bound that $\norm{\bxi^{(k)} }\leq \alpha_k$,
    and from \eqref{eq:xi^(c+1)}, we get,
    $\alpha_{c+1} \leq \sum_{i=0}^c \alpha_i \alpha_{c-i},$ 
    which forms the Catalan sequence.
\end{proof}

For the following discussions, we will maintain the normalization $\norm{\bm e_i}=1$.

We can now use the general expression in \cref{w_j-k} and find a bound,
\begin{equation}\label{wj-k-cartalan}
        \norm{\bm w_{j-k}}\leq \left( \frac{{\kappa(W)}\norm{F_2}}{\Delta} \right)^{k} \sum_{\overset{k_1+k_2+\cdots+k_r=k}{r= j-k}}  \frac{(2k_1)!}{(k_1+1)!k_1!} \frac{(2k_2)!}{(k_2+1)!k_2!} \cdots \frac{(2k_r)!}{(k_r+1)!k_r!}. 
    \end{equation}
Here $k_1, k_2, \cdots, k_r $ can be zero, in which case, the vector $\bm \xi$ coincides with an eigenvector $e_i$.

\begin{lemma}\label{lem: wj-k}
    The vector $\bm w_{j-k}$ in \cref{w_j-k} is bounded by,
    \begin{equation}
        \norm{\bm w_{j-k}}\leq \left( \frac{{\kappa(W)}\norm{F_2}}{\Delta} \right)^{k} {{j+k} \choose{k}} \frac{j-k}{j+k}. 
    \end{equation}
 \end{lemma}
\begin{proof}
    We first notice that the right hand side of \cref{wj-k-cartalan} is the sum of products of Cartalan numbers. Using the generating function, we have the bound. A detailed proof is provided in \cref{b-proof}.
\end{proof}

This will become the bound on the block of the matrix that contains the eigenvalue. 

\begin{theorem}
     Let $\mathbf{A}$ be the matrix from Carleman's linearization in \cref{matrix-car}. {Assume that $F_1=W\Lambda W^{-1}$ with $\Lambda$  being diagonal and containing the eigenvalues.}  Under the no-resonance condition \eqref{cond-res},  $\mathbf{A}$ is diagonalizable, $\mathbf{A} V = V D,$
where $D$ is diagonal, 
\begin{equation}
   D= \begin{bmatrix}
         \Lambda & & &  \\
           & \Lambda \otimes I + I \otimes \Lambda && \\
            && \Lambda \otimes I  \otimes I + I \otimes \Lambda  \otimes I +   I  \otimes I \otimes \Lambda & \\
             &&& \ddots 
    \end{bmatrix}.    
\end{equation}
In addition, the matrix $V$ that contains the eigenvectors can be partitioned in the same way as $\mathbf{A}$ as a block upper triangular matrix, with diagonal blocks being the identity matrices with the same dimension as those in  $\mathbf{A}$. Furthermore, the blocks  of $V$ has the following bound: $\forall k\leq j$
\begin{equation}\label{Vjk1}
    \norm{V_{k,j}}_{{1}} \leq \left({\kappa_1(W)} \frac{\norm{F_2}_1}{\Delta} \right)^{j-k} {{2j-k} \choose{j-k}}. 
\end{equation}

\end{theorem}

\begin{proof}
    The fact that when no resonance is present, $\mathbf{A}$ is diagonalizable has been noted in \cite{tsiligiannis1989normal}. For the eigenvalues obtained from each block $A_{j,j},$ since the eigenvectors in \cref{A_jj-eig} are linearly independent,  the resulting eigenvectors $\bm w$ constructed from previous Lemmas are linearly independent as well. Meanwhile, it is possible for two diagonal blocks, e.g., $A_{j,j}$ and $A_{m,m}$ for some $1<j<m$ to have common eigenvalues. But since the eigenvectors computed in the previous Lemma from $A_{j,j}$ have zero entries in the $m$th block, they are also linearly independent from those eigenvectors constructed from $A_{m,m}$. The rest of the Theorem follows from \cref{lem: wj-k}.
\end{proof}

{ Notice that we chose to use the $1$-norm in \cref{Vjk1}. This is because our previous analysis in \cref{lem: wj-k} showed a column-wise bound. Thus the matrix  $1$-norm provides a convenient bound without an explicit dimension dependence.   }
 A simple upper bound for this, using the inequality,
\( {n\choose k}  \leq \left(\frac{en}{k}\right)^k, \) 
is
\begin{equation}\label{bd-V_kj}
    \norm{V_{k,j}}_{{1}} \leq \left( \frac{2 e{\kappa_1(W)}\norm{ F_2}_{{1}}}{\Delta} \right)^{j-k}.
\end{equation}

{
\begin{corollary} \label{V inverse V}
 The matrix product $V_{k,k}^{-1}V_{k,j}$ has the following bound,
    \begin{equation}
        \norm{V_{k,k}^{-1}V_{k,j}}_1\le \left( \frac{2 e\kappa_1(W)\norm{F_2}_1}{\Delta} \right)^{j-k}.
    \end{equation}
\end{corollary}
\begin{proof}
Since $F_1W=\Lambda W$, we know 
\begin{equation}
    \begin{aligned}
    ((\l_{i_1}+\cdots+\l_{i_m})I-F_1)W &= G^{-1}W = W(((\l_{i_1}+\cdots+\l_{i_m})-\Lambda)\\
    \Longrightarrow W^{-1}G &= (((\l_{i_1}+\cdots+\l_{i_m})-\Lambda)^{-1}W^{-1}.
    \end{aligned}
\end{equation}
Since we chose the normalization $\norm{\bm e_i}=1$,
$\kappa_1(W)=\norm{W}_1\norm{W^{-1}}_1=\norm{W^{-1}}_1$. Therefore, for any matrix $G$ defined in \cref{Gmat},
\begin{equation}\label{W-1G}
    \norm{W^{-1}G}\le \frac{\kappa_1(W)}{\Delta}.
\end{equation}
Furthermore, in light of $V_{k,k}^{-1}=\left(W^{-1}\right)^{\otimes k}$, every column of $V_{k,k}^{-1}V_{k,j}$ takes the form of,
\begin{equation} \label{W inverse w}
    \left(W^{-1}\right)^{\otimes k}\bm w_{j-k}= \sum_{k_1+\cdots+k_j=k}\sum_{m_1,\cdots, m_j}\mathcal{L}\left((W^{-1})G_{p_1}\bxi_{m_1}^{(k_1)},\cdots,(W^{-1})G_{p_j}\bxi_{m_j}^{(k_j)}\right) .
\end{equation}
The implicit terms not shown in \eqref{W inverse w}, $W^{-1}\bm e_{i}$, is the standard unit vector with $1$ on $i$-th element. Since $W^{-1}G$ share same $\ell_1$ norm upper bound with $G$, we have
\begin{equation}
    \norm{V_{k,k}^{-1}V_{k,j}}_1=\norm{(W^{-1})^{\otimes k}\bm w_{j-k}}_1\le  \left( \frac{2 e\kappa_1(W)\norm{F_2}_1}{\Delta} \right)^{j-k}.
\end{equation}

\end{proof}
It is worthwhile to emphasize that a direct estimate for \cref{bd-V_kj} using multiplicative property will incur another $\kappa_1(W)^k$ term, which will cause a significant over estimation, especially when $k$ is close to $N$.  
}

We now proceed to estimate $\bm \zeta$ in \cref{eta-zeta}. Recall that $V$ is an upper tringular block matrix and
\begin{equation}
     \bm \zeta(t):=  V^{-1}  \mat{0\\  \vdots \\ 0 \\ A_{N,N+1} \bm x^{\otimes(N+1)}(t)} = \mat{ \bm \zeta_1 \\ \vdots \\  \bm \zeta_{N-1}\\ \bm  \zeta_N}.
\end{equation}

{
\begin{lemma}
     $\bm \zeta$ can be partitioned into blocks $\bm \zeta=(\bm \zeta_1, \cdots,  \bm \zeta_N)^T$ according to the Carleman linearization,
    and the blocks obey the following the bounds
    \begin{equation}
        \norm{\bm \zeta_j}_1 \leq  \left( \frac{4 e\kappa_1(W)\norm{ F_2}_1}{\Delta} \right)^{N-j}\norm{\bm\zeta_N}_1.
    \end{equation}
\end{lemma}
\begin{proof}
    By back substitution and the triangle inequality, we have,
    \begin{equation}
       \norm{\bm \zeta_j}_1 \le \sum_{j< j_1 <j_2 < \cdots < j_r< N } \norm{V_{j,j}^{-1}V_{j,j_1}V_{j_1,j_1}^{-1}V_{j_1,j_2}V_{j_2,j_2}^{-1} \cdots V_{j_{r},\ell} \bm  \zeta_N}_1.
    \end{equation}
    Furthermore, we observe that 
    \begin{equation}\label{summation subset}
        \sum_{j< j_1 <j_2 < \cdots < j_r< N } 1 = 2^{N-j-1}-1\le 2^{N-j}.
    \end{equation}
    This follows from the fact that $\sum_{j< j_1 <j_2 < \cdots < j_r< N } 1$ is equivalent to the number of subsets of $\{j+1,\cdots,N-1\}$ excluding the empty set. 
    Using the bound in \cref{V inverse V} and \eqref{summation subset}, we find that
    \begin{equation}
        \norm{\bm \zeta_j}_1\le \sum_{j< j_1 <j_2 < \cdots < j_r< N}\left(\frac{2e\kappa_1(W)\norm{F_2}_1}{\Delta}\right)^{N-j}\norm{\bm\zeta_N}_1\le \left(\frac{4e\kappa_1(W)\norm{F_2}_1}{\Delta}\right)^{N-j}\norm{\bm\zeta_N}_1.
    \end{equation}
    Notice that
    \begin{equation}
        \bm\zeta_N = V_{N,N}^{-1}A_{N,N+1}\bm x^{\otimes(N+1)}. 
    \end{equation}
    We can then bound  $\bm\zeta_N$ by
    \begin{equation}
        \norm{\bm\zeta_N}_1\le N\kappa_1(W)^N\norm{F_2}_1 \mu^{N+1}. 
    \end{equation}
    
\end{proof}

}

As a result, $\forall t\geq 0,$ we have from \eqref{eq:zeta2eta},  {
\begin{equation}
 \begin{aligned}
         \norm{\bm \eta_1(t)}_1 &\leq  \int_0^t\left(\norm{V_{1,1}}_1\norm{\bm \zeta_1}_1 +  \norm{V_{1,2}}_1\norm{\bm \zeta_2}_1 + \cdots +   \norm{V_{1,N}}_1\norm{\bm \zeta_N}_1\right)d\tau\\
         & \leq t\left( \sum_{j=1}^{N-1}  \left(\frac{2e\kappa_1(W)\norm{F_2}_1}{\Delta}\right)^{j-1}\left(\frac{4e\kappa_1(W)\norm{F_2}_1}{\Delta}\right)^{N-j}\norm{\bm\zeta_N}_1+\norm{V_{1,N}}_1\norm{\bm\zeta_N}_1\right)\\
         &\leq t\sum_{j=1}^N \left(\frac{4e\kappa_1(W)\norm{F_2}_1}{\Delta}\right)^{N-1}\norm{\bm\zeta_N}_1\\
         &\le tN\kappa_1(W)\norm{F_2}_1\mu^2\left(\frac{4e\kappa_1(W)^2\norm{F_2}_1\mu}{\Delta}\right)^{N-1}.
 \end{aligned}
\end{equation}
}

Collecting these results, we state  the main theorem, {
\begin{theorem}
     Assume that the solution of the nonlinear ODE {\cref{eq:1.1}} satisfies a uniform bound,
    \begin{equation}
        \norm{\bm x(t)}_1 \leq \mu, \; \forall t\in [0,T],
    \end{equation}
{and $F_1$ satisfies \cref{assump1}}. Then the error in the Carleman linearization can be bounded by
\begin{equation}\label{1-norm-bound}
    \norm{\bm x(T) - \bm y_1(T) }_1 \leq N CT \mathsf{R}_r^{N-1},
\end{equation}
where 
\begin{equation}
    C\coloneqq\kappa_1(W)\norm{F_2}_1\mu^2,\quad \mathsf{R}_r:= \frac{ 4e\mu \kappa_1(W)\norm{F_2}_1}{\Delta}.
\end{equation}
\end{theorem}
Notice that since $\norm{\bullet}_2 \leq \norm{\bullet}_1$, the bound in \cref{1-norm-bound} implies that $\norm{\bm x(T) - \bm y_1(T) }_2\leq N CT \mathsf{R}_r^{N-1}.$
}

\begin{remark}
    For Hamiltonian systems, the eigenvalues of $F_1$ are purely imaginary, i.e., $\pm i\omega_j $. Then no-resonance condition \eqref{cond-res} implies that any two diagonal blocks, $A_{j,j}$ and $A_{k,k}$ with $k>j$, do not share eigenvalues. Therefore, it is not necessary to involve Fredholm alternative condition in the proof.
\end{remark}

\begin{remark}\label{rem:PBC}
Some ODEs come with zero eigenvalues that correspond to translational symmetry or the application of the periodic boundary conditions. This implies that $\bm v^T \bm x(t)$ is a first integral for some vector $\bm v$, implying that,
\begin{equation}\label{first-integral}
    \bm v\in \left( \text{Range}(F_1)\right)^\perp, \quad \bm v\in \left( \text{Range}(F_2)\right)^\perp.
\end{equation}
   
Although the presence of the zero eigenvalues obviously leads to a resonance, i.e., $\lambda_i = \lambda_i + m\times 0$ for any $\lambda_i$ and any positive integer $m$, these resonance modes can be easily separated, { assuming we know the vector $\bm v$, which is often the case.}  For example, when $\sum_{i=1}^n x_i $ is a first integral, then one can apply a coordinate transformation and reduce the ODEs to an $(n-1)$-dimensional ODE system, while still in the form \eqref{eq:1.1}, i.e.,
\begin{equation}\label{ode-x-tilde}
    \frac{d}{dt}  \widetilde{\bm x} = \widetilde{F}_1 \widetilde{\bm x} + \widetilde{F}_2 \widetilde{\bm x} \otimes \widetilde{\bm x}.
\end{equation}
In addition, one can use \eqref{first-integral} and show that the Carleman system \cref{y-car} is exactly equivalent to the Carleman system derived from \cref{ode-x-tilde}.
As a result, such zero eigenvalues can be removed from the error analysis.
\end{remark}

\begin{remark}
    When the no-resonance condition \eqref{cond-res} is violated, the analysis has to be modified accordingly. First, when the resonance occurs at an order higher than $N$, i.e., $\sum m_j >N$, then the current analysis still applies straightforwardly.
    Secondly, the presence of the resonant modes within the truncation order $N$ might introduce repeated eigenvalues with geometric multiplicity greater than $m>1$. As a result, the matrix exponential in \cref{eta-zeta} will introduce a polynomial growth $O(t^m)$. Therefore, if $m \ll N$, {  it is still possible to get a convergence that only has logarithmic dependence on $T$. }   
\end{remark}

\section{Numerical Results}\label{sec: num}

\subsection{Carleman linearization for { viscous} Burgers' equation}

Our first numerical test is  performed on the Burgers' equation,  a standard example for a nonlinear PDE
\begin{align}
     u_t + cu u_x   = u_{xx}, \quad x \in [-\frac12, \frac12 ].  \label{bphyequ}
\end{align}
We introduce a parameter $c>0$ here to test the effect of the nonlinearity. By assuming a homogeneous Dirichlet boundary condition, and applying an upwind scheme (e.g., when $u\geq 0$) \cite{leveque1992numerical}, we can reduce the PDE to a nonlinear ODE system, with grid size  $\Delta x = \frac{1}{n}$, 
\begin{equation} \label{bnumcal}
    \begin{aligned}
    \frac{d}{dt}{u}_{j}=& -c\frac{u_{j}^{2}-u_{j-1}^{2}}{2\Delta x} + \frac{u_{j+1}-2u_{j}+u_{j-1}}{\Delta x^{2}}, \quad j=1, 2, \cdots, n  \quad 
u_{0}=&u_{n+1}=0.
\end{aligned}
\end{equation}

As a result of this semi-discrete approximation, we arrive at a nonlinear ODE system in the form of \eqref{eq:1.1}. In particular, $F_1$ is the standard tri-diagonal matrix, {with resonance paramteter $\Delta = 0.1497$}. In addition, $F_2$  is 1-sparse with only the following non-zero entry.

\[ \left(F_{2}\right)_{i,(i,i)}=-\frac{c}{2\Delta x}\quad \text{for all }1\le i\le n, \quad  \left(F_{2}\right)_{i, (i-1,i-1)}=\frac{c}{2\Delta x}\quad \text{for all }2\le i\le n. \]

In the numerical tests, we initialize the Burgers' equation by a bell-shaped profile,
\begin{align}
 u(x,0)=-\tan^{-1}{\left(20(x-\frac{1}{4})\right)}+\tan^{-1}{\left(20(x+\frac{1}{4})\right)}. \label{sim_initial}
 \end{align}
The solution at later times is displayed in \cref{BUR_mesh}, which shows decay due to the smoothing from the viscosity term. 
\begin{figure}[htbp]
	\centering
		\includegraphics[width=0.56\linewidth]{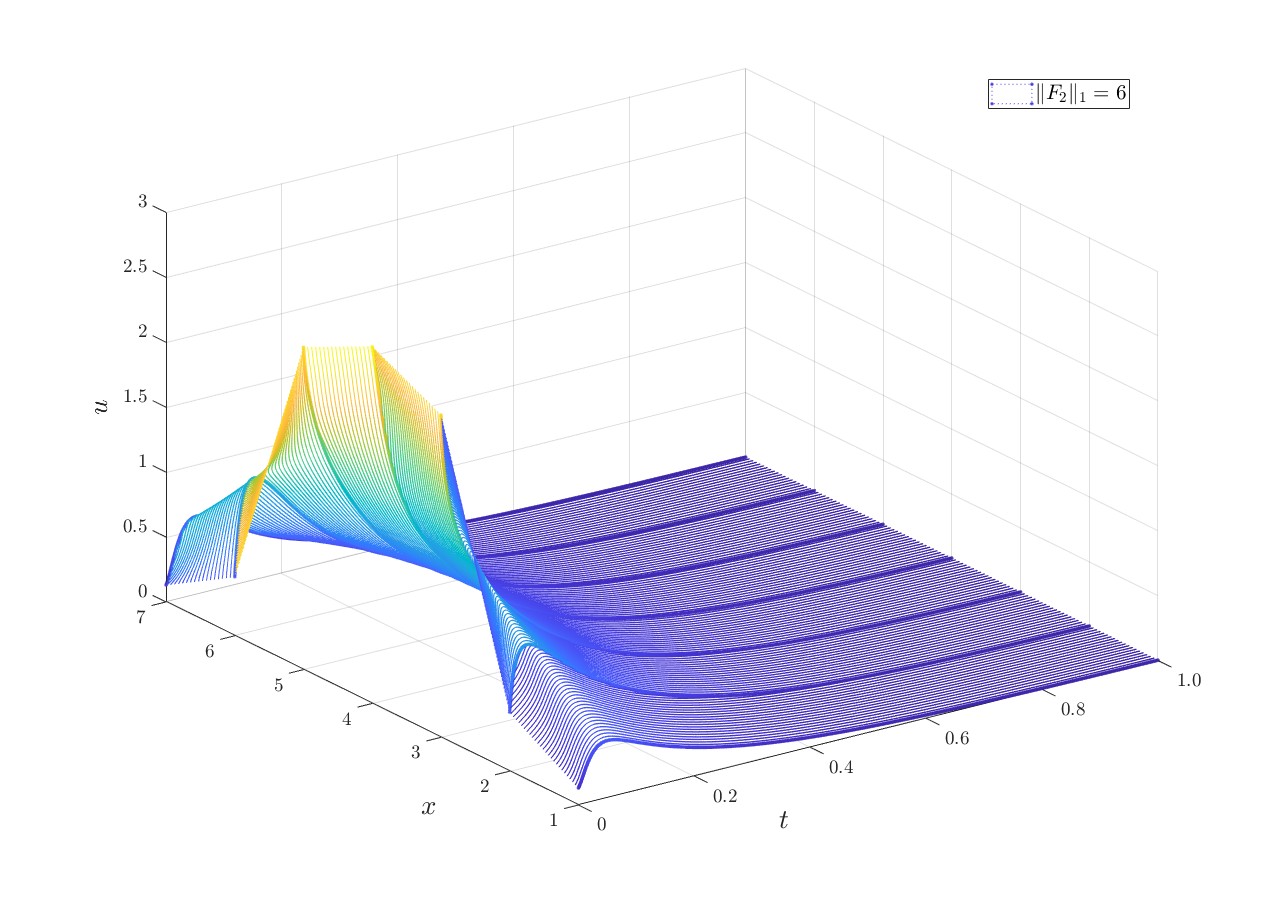}
		\caption{Solution of the Burgers' equation \eqref{bphyequ} in the time interval $[0,1]$. The solutions are obtained using the 4th-order Runge-Kutta methods applied to the nonlinear ODEs \eqref{bnumcal}. }
		\label{BUR_mesh}
\end{figure}

We first pick $n=7,$ which is a relatively small system compared to the numerical tests in \cite{liu2021efficient}. But we extend the Carleman linearization \eqref{matrix-car} to much higher levels up to $N=8$.  Our first test focuses on the effect of the level of nonlinearity, indicated by $\norm{F_{2}}$. This is confirmed by the numerical results in \cref{BURF2}: weaker nonlinearity induces smaller truncation error, and when the norm of ${F_{2}}$ is too large, the approximation from the Carleman linearization begins to diverge. For the setup of our numerical experiments, the parameter $\mathsf{R}_d\approx0.8223\norm{F_{2}}_{2}$ {and $\mathsf{R}_r\approx 1261\norm{F_2}_{1}$} can be much greater than 1 in the numerical experiments.  Nevertheless, the Carleman linearization still converges, which was also observed in  \cite{liu2021efficient}. This could be attributed to the smoothing effect, which can be seen from \cref{BUR_mesh}: the solution decays over time, which in light of \eqref{err-eta1}, can make the error much smaller. Another option to reduce the parameter $\mathsf{R}_d$ is to adjust the real part of the eigenvalues of $F_1$. Toward this end, we shift the eigenvalues of $F_1$ toward zero by adding $\beta I$ to $F_1$. By choosing $\beta=5$, the eigenvalue with the smallest absolute value is $\lambda=-0.48$s  {The corresponding resonance parameter $\Delta=0.4287$ and $\mathsf{R}_r\approx 440.1\norm{F_2}_{1}$.} We can see from the right panel in \cref{BURF2} that the  error from Carleman linearization begins to diverge when $\norm{F_2}_{1}=12$. On one hand, this highlights the importance of the dissipative condition \eqref{eq: dissip}. On the other hand, it also shows that Carleman linearization still converges when an eigenvalue is very close to zero. 

\begin{figure}[htbp]
\centering
		\includegraphics[width=0.46\linewidth]{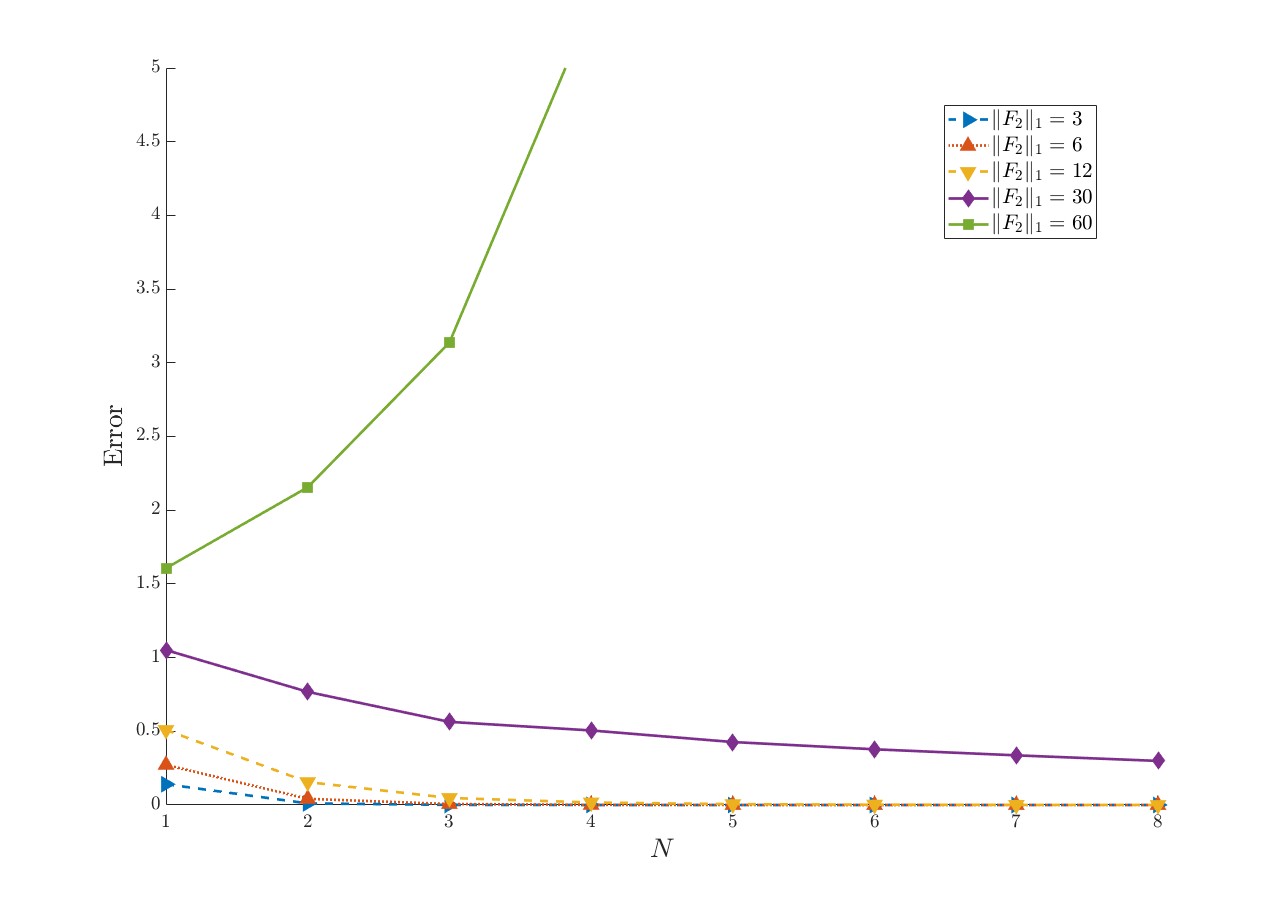}
        \includegraphics[width=0.46\linewidth]{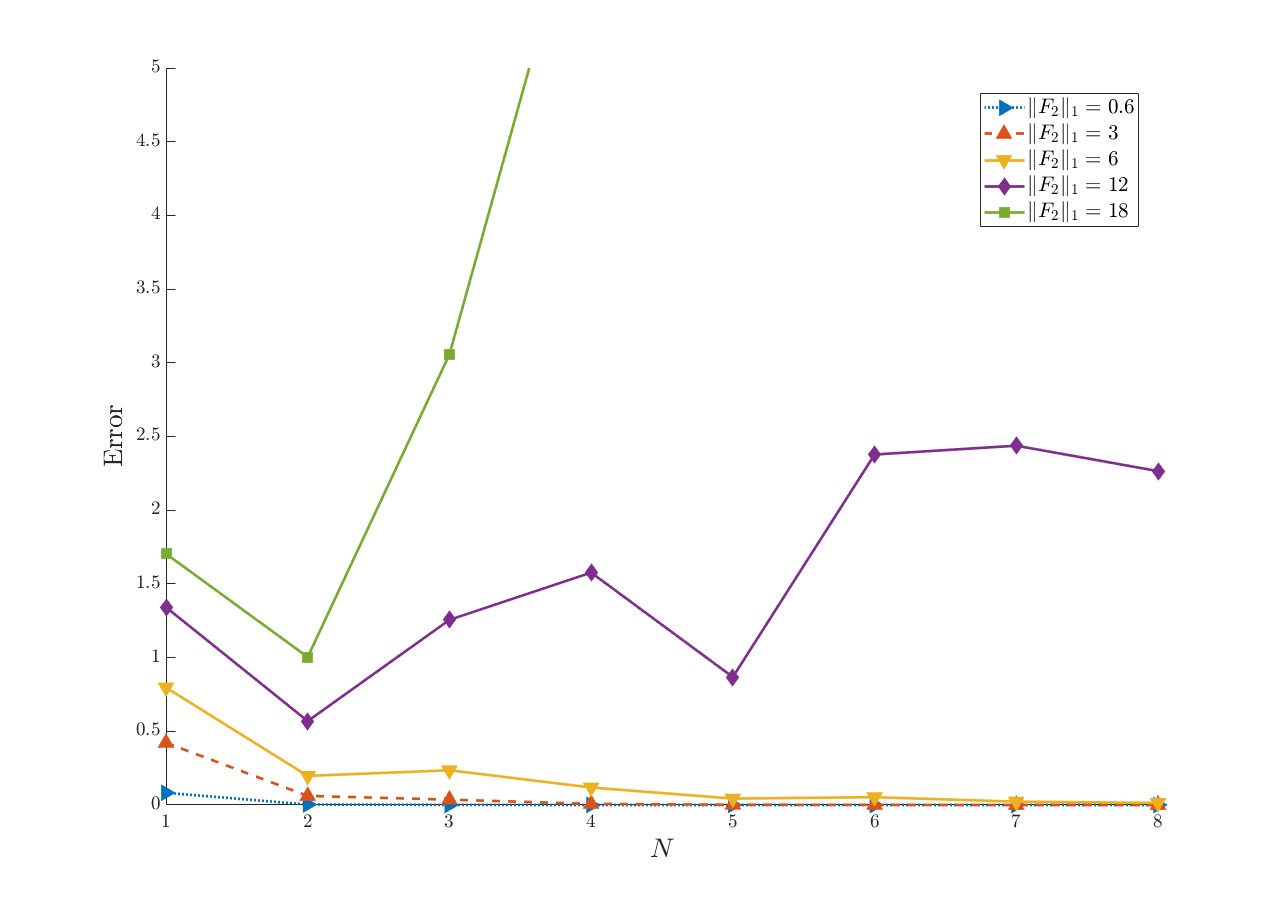}
		\caption{Left: The Carleman linearization error for the Burgers' equation in the time interval $[0,1]$  for different trunction level $N \leq 8$  with different level of nonlinearity, indicated by $\norm{F_{2}}_{1}$, which is controlled by changing $c$ in \cref{bnumcal}; Right:  We add a  $\beta I$  term to $F_1$ in \cref{bnumcal} so that the leading eigenvalue is close to zero. }
		\label{BURF2}
	\end{figure}

The analysis in \cite{forets2017explicit,liu2021efficient} (also see \cref{err-eta1}) suggests a small error for a short time $T$. To test this, we fix  $\norm{F_{2}}_{1}=6$ and examine the error of Carleman linearization using $L^\infty$ norm in time. 
We observe from \cref{BURT} that for shorter times, the error is indeed small. However, the error begins to saturate for longer time intervals. 

\begin{figure}[htbp]
		\centering
		\includegraphics[width=0.56\linewidth]{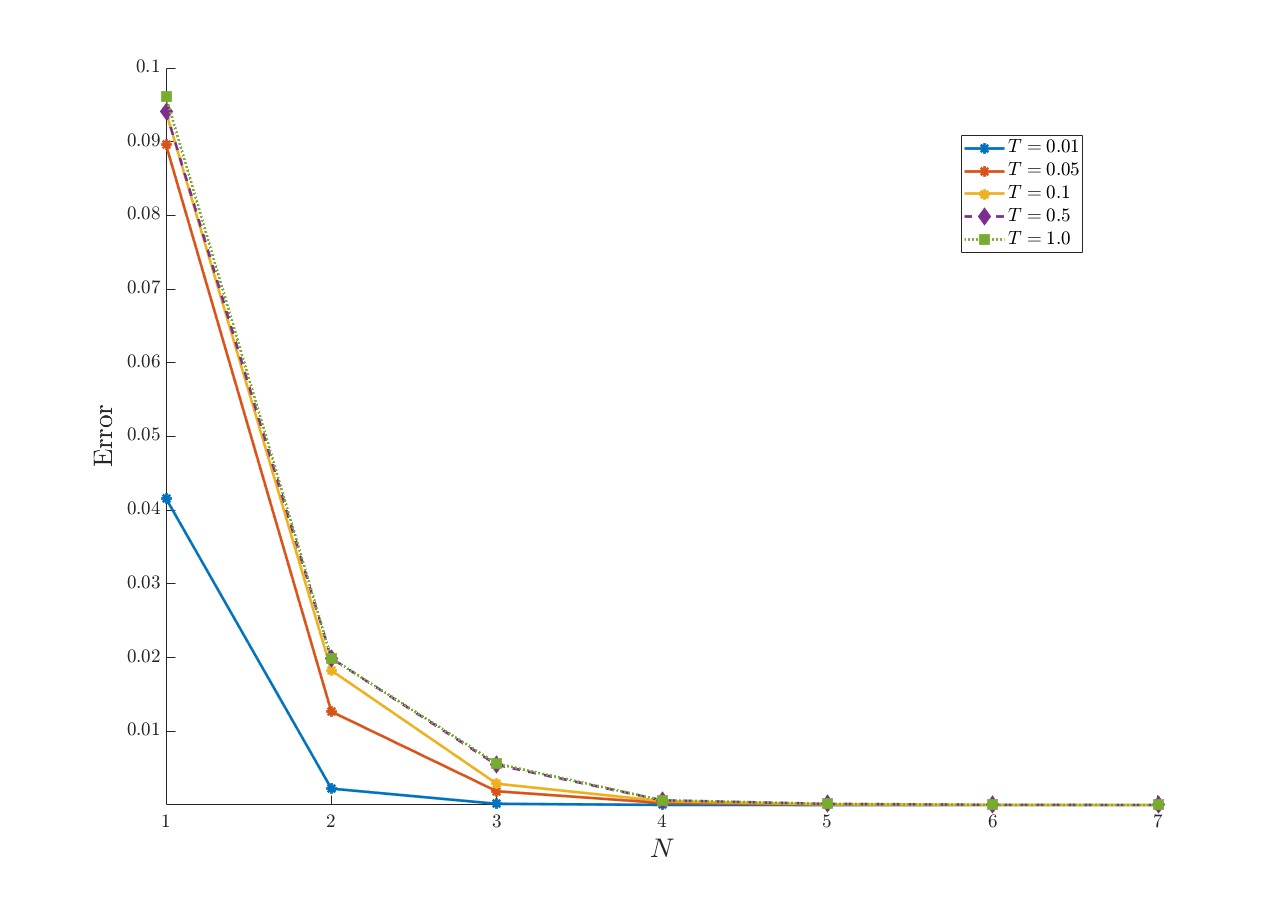}
		\caption{The Carleman linearization error for the Burgers' equation for different truncation levels  $N$ and different time durations $T$. Here we set $\norm{F_{2}}_{1}=6$.
  }		\label{BURT}
	\end{figure}

Finally, we consider a larger ODE system from \eqref{bphyequ} using a finer spatial discretization with $n=31$. The corresponding parameter is $\mathsf{R}_d=1.145\norm{F_{2}}_2$ { and $\Delta = 1.802 \times 10^{-2},\mathsf{R}_r\approx 42920\norm{F_2}_{1}$}. By comparing the previous results in \cref{BURF2} with discretization $n=7$ with results shown  \cref{BURF2_modnx}, we can see the point for $\norm{F_{2}}$ when the error begins to diverge is larger.

 \begin{figure}[htbp]
\centering
		\includegraphics[width=0.56\linewidth]{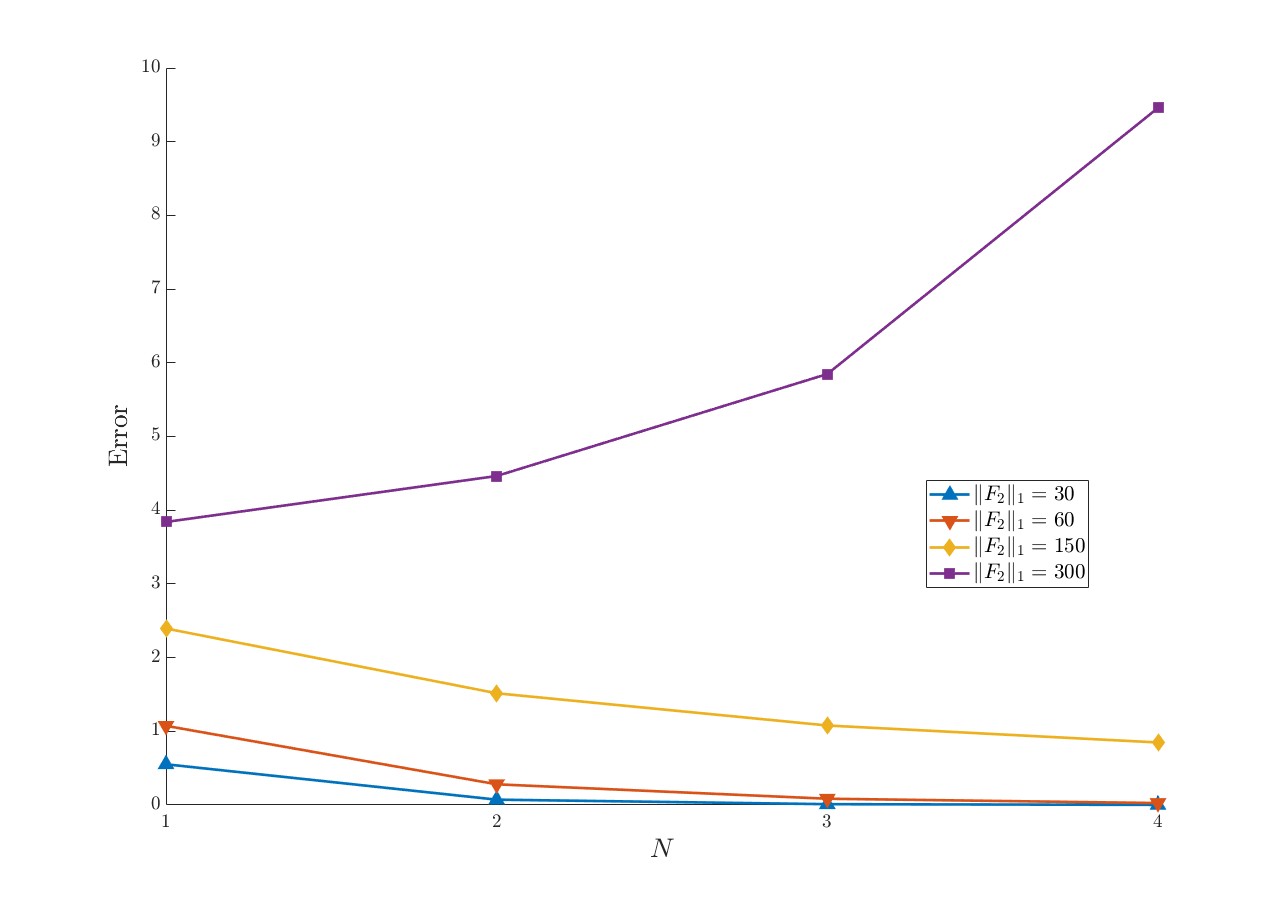}
		\caption{The Carleman linearization error for the Burgers' equation in the time interval $[0,0.5]$  for different trunction level $N \leq 4$  with different level of nonlinearity, indicated by $\norm{F_2}$, which is controlled by changing $c$ in \cref{bnumcal}. Here we set $n=31$. }
		\label{BURF2_modnx}
	\end{figure}

\subsection{Carleman linearization for the KdV equation} 

In this section, we provide numerical results from simulating the KdV equation, 
\begin{align}
    u_t + cu u_x  = -u_{xxx}, \quad x\in (-\frac12, \frac12), \quad t>0.  \label{kdvequ}
\end{align}
In contrast to the dissipation in the Burgers' equation \cref{bphyequ}, the KdV equation involves a third-order derivative, introducing a dispersive regularization \cite{lax1983small}.  To convert the PDE to the nonlinear ODE system,  we use the following discretization,
\begin{align}
    \frac{d}{dt} {u}_{j}= -c\frac{u_{j}^{2}-u_{j-1}^{2}}{2\Delta x} - \frac{u_{j+2}-2u_{j+1}+2u_{j-1}-u_{j-2}}{2\Delta x^{3}}.\label{knumcal}
\end{align}

We employ periodic boundary conditions. Notice that \cref{knumcal} is no longer dissipative. The eigenvalues of $F_{1}$ are purely imaginary and the criterion using $\mathsf{R}_d$ in \cite{liu2021efficient} does not apply.  {In light of \cref{rem:PBC}, we neglected the zero eigenvalues when computing the resonance parameter $\Delta $ in \cref{Delta}, and the resulting value is $\Delta = 9.860$.}

Similar to the setup of the numerical experiments for the Burgers' equation, we begin by selecting a modest system size of $n=7$, focusing on the Carleman linearization \eqref{matrix-car} up to level $N=8$. The solution to the KdV equation, initialized via \cref{sim_initial}, is depicted in \cref{KDV_Mesh}, where one can observe a soliton-like propagation through the system within the time interval $[0,0.1]$ and under the periodic boundary condition. In particular, unlike the solution of the Burgers' equation in \cref{BUR_mesh}, the solution of the KdV equation does not exhibit a decay. 

\begin{figure}[htbp]
	\centering
		\includegraphics[width=0.56\linewidth]{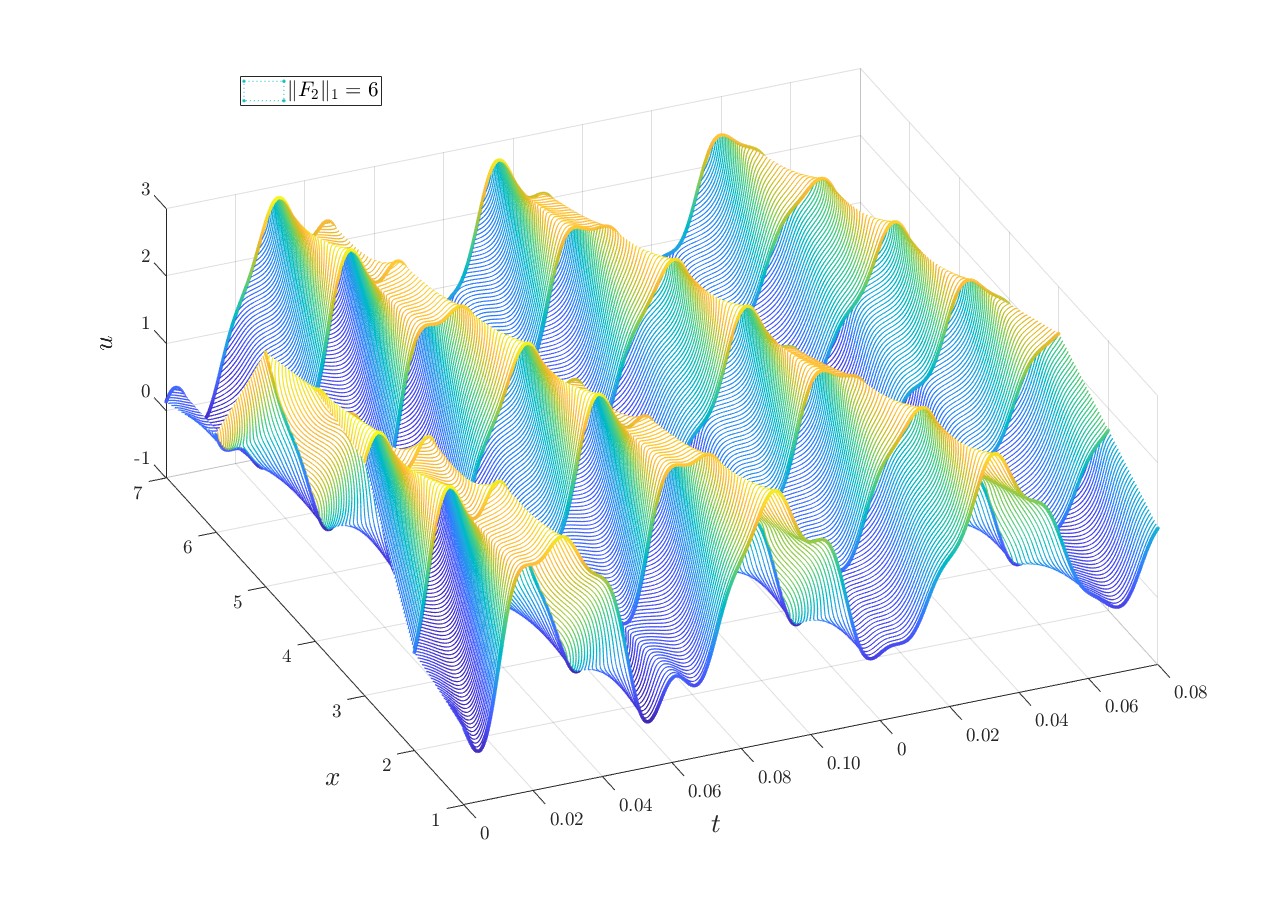}
		\caption{Solution of the KdV \cref{kdvequ} in the time interval $[0,0.1]$. The solutions are obtained by using the 4th-order Runge-Kutta methods applied to the nonlinear ODEs \eqref{knumcal}. } 
		\label{KDV_Mesh}
\end{figure}

Next, we investigate the error due to the nonlinearity for various simulation durations, quantified by the $L^\infty$ norm over time. By tuning the parameter $c$ in \cref{kdvequ}, we can adjust the level of the nonlinearity, and the results are illustrated in \cref{KDVF2}. It is clear that even without the dissipative condition \eqref{eq: dissip}, the Carleman linearization still convergences when $\norm{F_{2}}$ is sufficiently small. To compare to the analysis in \cref{thm:main}, we calculated the parameter {$\mathsf{R}_r=21.20\norm{F_{2}}_{1}$}. Due to the lack of dissipation, we attribute the observed convergence to the resonance condition \eqref{cond-res}. It is also interesting the convergence can still occur when $\mathsf{R}_r>1.$

\begin{figure}[htbp]
		\centering
		\includegraphics[width=0.465\linewidth]{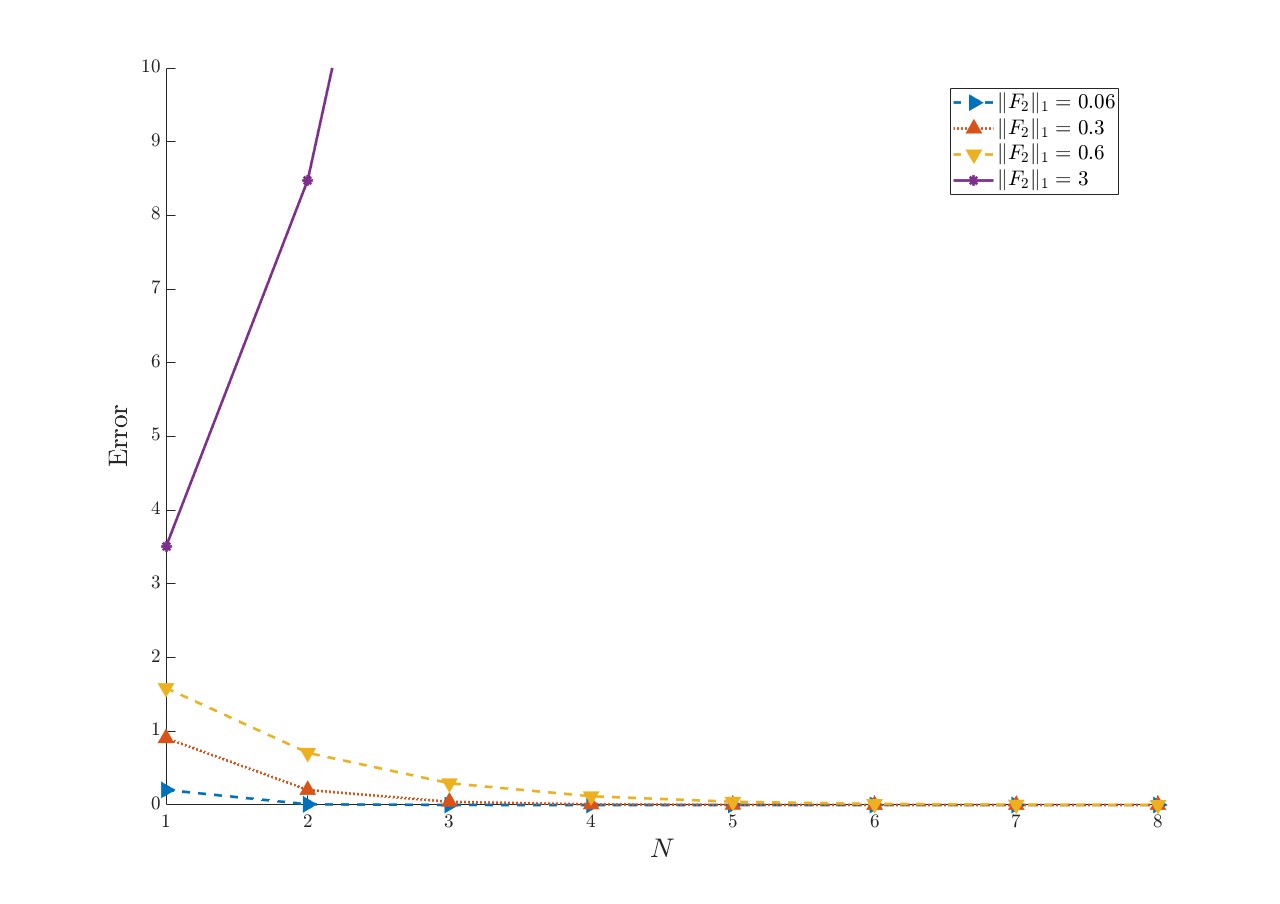}
  \includegraphics[width=0.465\linewidth]{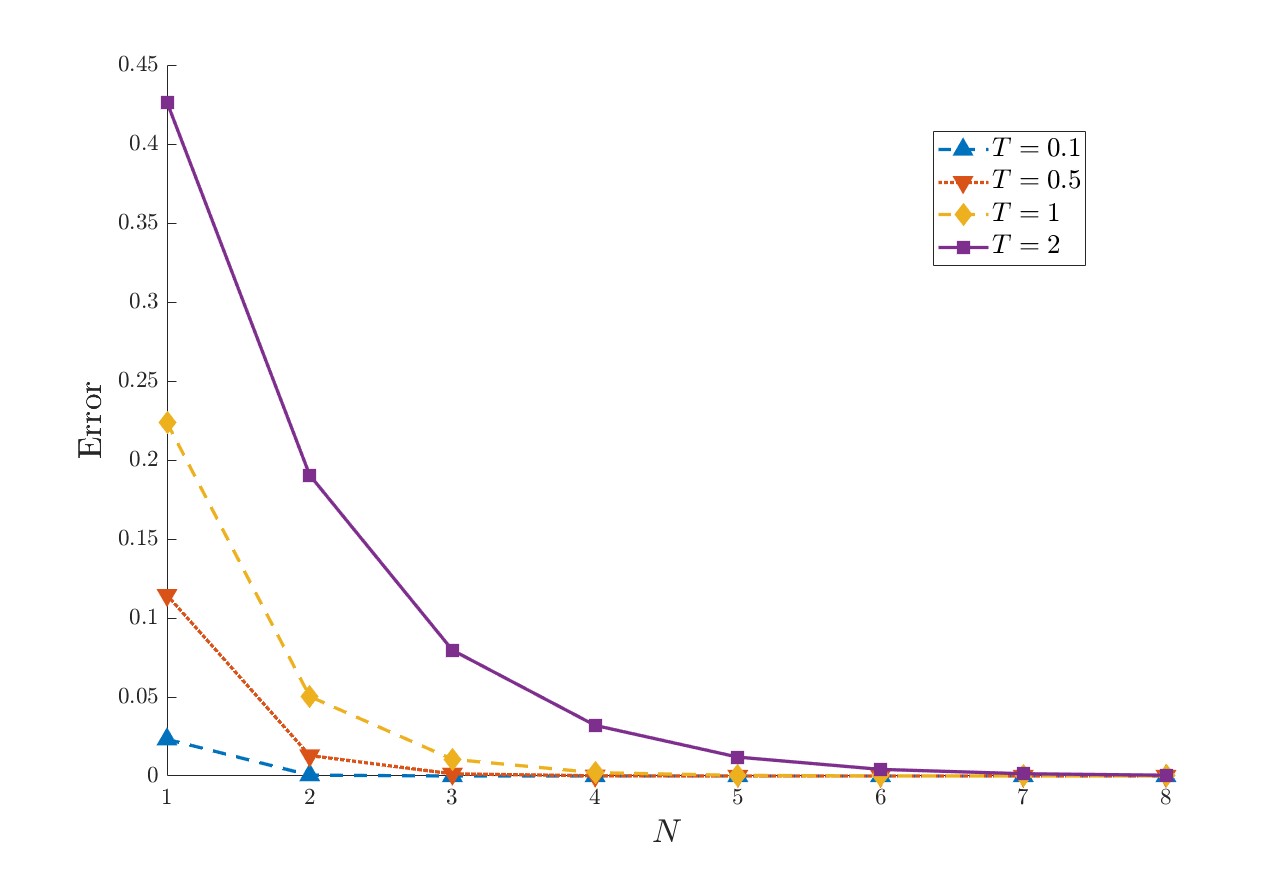}
		\caption{Left: The Carleman linearization error for the KdV equation in the time interval $[0,0.1]$  for different truncation level $N \leq 8$  with different level of nonlinearity, indicated by $\norm{F_2}_1$, which is controlled by changing $c$ in \cref{knumcal}; Right: The Carleman linearization error for the KdV equation for different truncation level $N$ and different time durations $T$. Here we set $\norm{F_{2}}_{1}=0.6$.  }
		\label{KDVF2}	
\end{figure}

To examine the impact of total duration $T$, we fixed $\norm {F_{2}}_{1}=0.6$. We observe from the right panel in \cref{KDVF2} that the error grows with total duration, while the convergence with truncation level is still markedly evident.

\subsection{Carleman linearization for the Fermi-Pasta-Ulam (FPU) dynamics }

Our next example is the Fermi-Pasta-Ulam (FPU) chain model, which is a foundational example in nonlinear dynamics and statistical mechanics for the study of chaos \cite{rink2001symmetry}, thermodynamic properties \cite{cretegny1998localization} and phonon transport \cite{dhar2008heat}. FPU chain model consists of a one-dimensional chain of particles connected by linear and nonlinear springs. For $i=1,\cdots, p$, the governing equations for the atomic displacement, denoted by $u_i$, are given by
\begin{equation}\label{eq: fpu}
    \frac{\mathrm{d}^{2} {u}_{i}}{\mathrm{d} t^{2}}  = k \left(u_{i+1}-2u_{i}+u_{i-1}\right) + \alpha \left(u_{i+1}-u_{i}\right)^{3} -\alpha \left(u_{i}-u_{i-1}\right)^{3}.
\end{equation}
      
Here we set the mass $m=1$ and impose Dirichlet boundary conditions $u_{0}=u_{p+1}=0$. To express the ODEs in a first-order form, we define
\begin{equation}
    \bm x=\left(u_{1},u_{2},\dots,u_{p-1},u_{p},\;\dot{u}_{1},\dot{u}_{2},\dots,\dot{u}_{p-1},\dot{u}_{p}\right)^{T},
\end{equation}
where $p$ is the number of moving particles.  Due to the linear and cubic terms in  \cref{eq: fpu}, one can rewrite the equations of motion as follows,
\begin{equation}\label{uequ}
      \frac{d}{dt} \bm x= F_{1} \bm x + F_{3} \bm x\otimes \bm x\otimes \bm x. 
\end{equation}
Here $F_{1} \in \R^{2p\times 2p}$  is a sparse matrix with nonzero entries:  $ \left(F_{1}\right)_{i, i}=1, \left(F_{1}\right)_{i+p,i}=-2k, \forall 1\le i\le p, $ 
$\left(F_{1}\right)_{i+p, i+1}=k, \quad \forall 1\le i\le p-1$, and $ \left(F_{1}\right)_{i+p, i-1}=k, \quad  \forall 2\le i\le p. $

Notice that $F_1$ has imaginary eigenvalues that are proportional to the frequency of the chain. In addition,  we removed the zero frequency by using Dirichlet boundary conditions. 

Note that $F_{3} \in \R^{2p\times \left(2p\right)^{3}}$, we denote $(i,j,k)$ to be the column of $F_{3}$ acting on $\bm x_{i}\otimes \bm x_{j}\otimes \bm x_{k}$, then $F_{3}$ is sparse with the following non-zero elements:
\( \left(F_{3}\right)_{i+p,(i,i,i)}=-2\alpha, \quad \text{for any } 1\le i\le p. \)
Meanwhile, for any $1\le i\le p-1$, they are given by,
\[ \case{ \left(F_{3}\right)_{i+p, (i+1,i+1,i+1)}=\left(F_{3}\right)_{i+p, (i,i,i+1)}=\left(F_{3}\right)_{i+p, (i,i+1,i)}=\left(F_{3}\right)_{i+p, (i+1,i,i)}=\alpha, \\ 
          \left(F_{3}\right)_{i+p,(i+1,i+1,i)}=\left(F_{3}\right)_{i+p,(i,i+1,i+1)}=\left(F_{3}\right)_{i+p,(i+1,i,i+1)}=-\alpha. }\]
Similarly, for any $2\le i\le p$,
\[ \case{\left(F_{3}\right)_{i+p, (i-1,i-1,i-1)}=\left(F_{3}\right)_{i+p, (i,i,i-1)}=\left(F_{3}\right)_{i+p, (i,i-1,i)}=\left(F_{3}\right)_{i+p, (i-1,i,i)}=\alpha, \\ 
     \left(F_{3}\right)_{i+p,(i-1,i-1,i)}=\left(F_{3}\right)_{i+p,(i,i-1,i-1)}=\left(F_{3}\right)_{i+p,(i-1,i,i-1)}=-\alpha.} \]

This is constructed from the condition that $\left(F_{3} \bm x\otimes  \bm x\otimes \bm x\right)_{p+i} = \alpha (u_{i+1}-u_{i})^{3}-\alpha (u_{i}-u_{i-1})^{3}.$  

With the model fully described, we write down the Carleman linearization for equation \eqref{uequ}
\begin{align}\label{fpu-car}
    \frac{\mathrm{d}}{\mathrm{d} t}\begin{pmatrix}
 \bm y_{1}\\
\bm  y_{2}\\
\bm  y_{3}\\
\bm  y_{4}\\
 \vdots \\
\bm y_{N}
\end{pmatrix}=\begin{pmatrix}
A_{1,1}  & 0 & A_{1,3}& 0 & \cdots & 0 & 0 \\
0  & A_{2,2} & 0 &  A_{2,4}& \cdots & 0 & 0 \\
 0 & 0 & A_{3,3} &  0& \ddots & \vdots & \vdots \\
 0 & 0 & &A_{4,4}  & \ddots & \vdots & \vdots  \\
  \vdots  & \vdots  &  \vdots &\vdots  & \ddots & \vdots & \vdots  \\
 0 & 0 & 0 & 0 & \cdots & A_{N-1,N-1} & 0  \\
 0 & 0 & 0 & 0& \cdots  & 0 & 
A_{N,N}\end{pmatrix}\begin{pmatrix}
\bm  y_{1}\\
\bm  y_{2}\\
\bm  y_{3}\\
\bm  y_{4}\\
 \vdots \\
\bm y_{N}
\end{pmatrix}.
\end{align}
where $A_{j,j}=\sum_{i=1}^{j-1}I^{\otimes i}\otimes F_{1}\otimes I^{\otimes j-1-i}$ and $A_{j,j+2}=\sum_{i=1}^{j-1}I^{\otimes i}\otimes F_{3}\otimes I^{\otimes j-1-i}$. Due to the cubic nonlinearity, the off-diagonal blocks are shifted toward the upper right corner. In the following, we show that this system can be converted into our standard form \cref{eq:1.1} when $N$ is even.

\begin{figure}[htbp]
	\centering
		\includegraphics[width=0.56\linewidth]{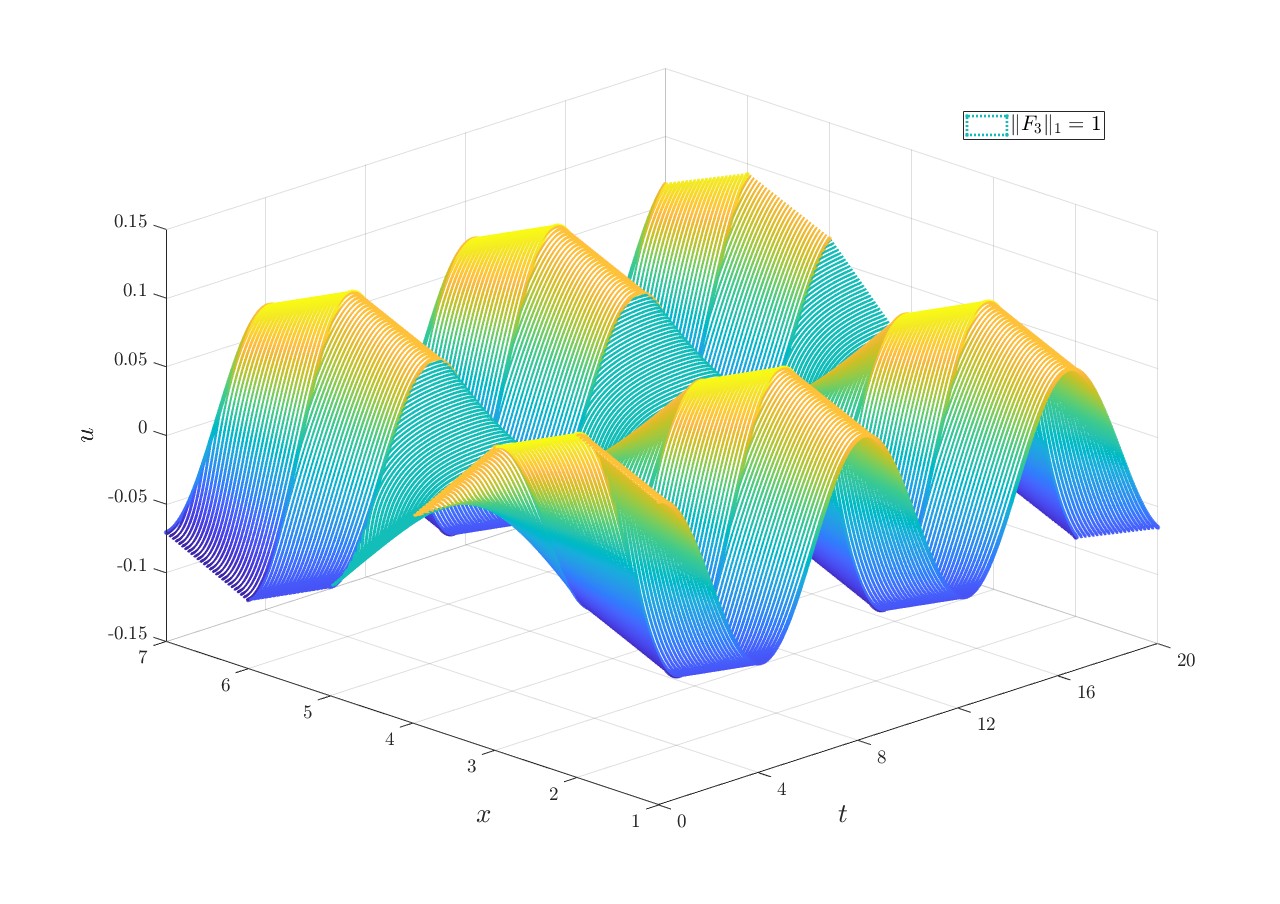}
		\caption{The displacement part of the solution of the FPU chain model \eqref{eq: fpu} in the time interval $[0,20]$. }
		\label{FPU_Mesh}
\end{figure}

\begin{figure}[htbp]
		\centering
		\includegraphics[width=0.46\linewidth]{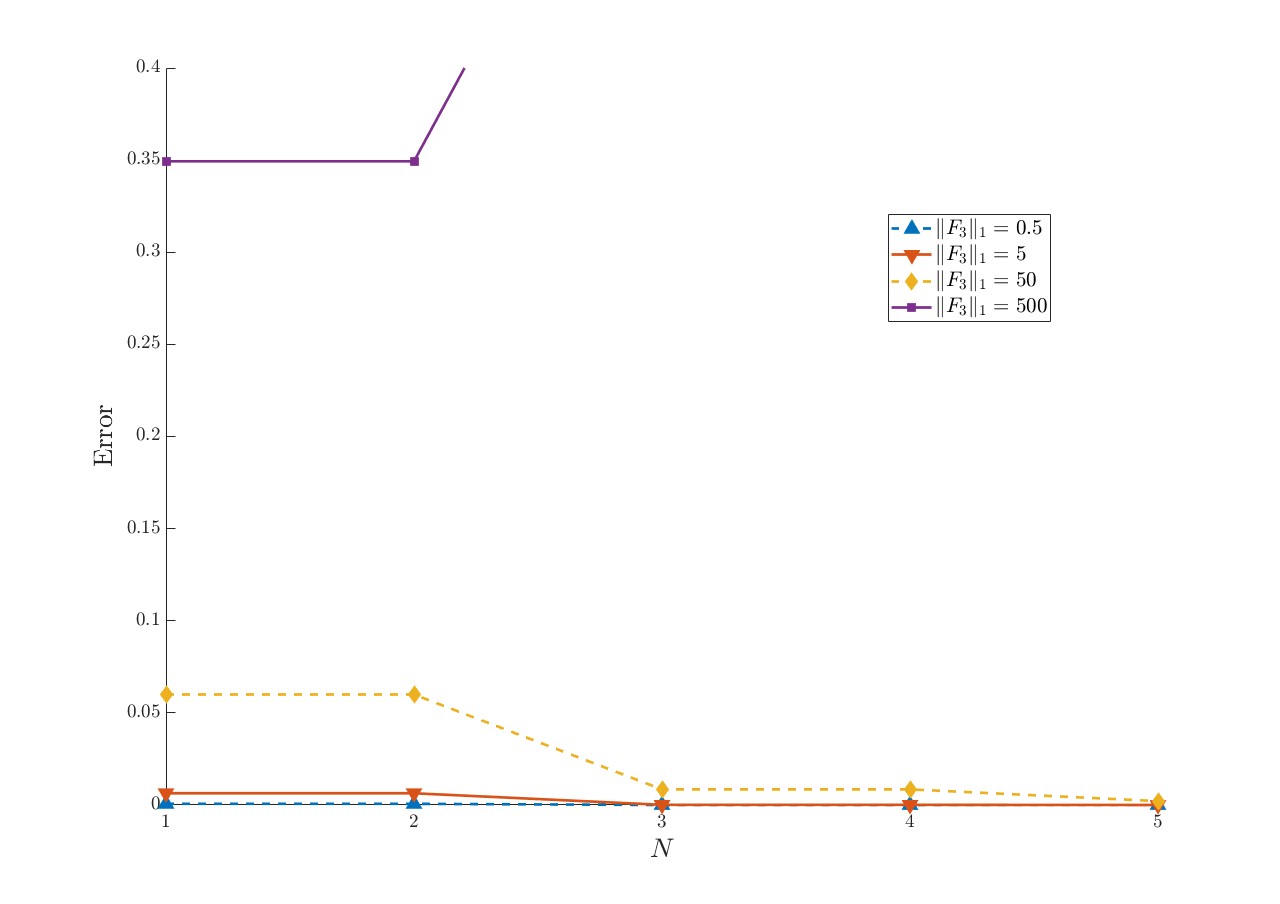}
  \includegraphics[width=0.46\linewidth]{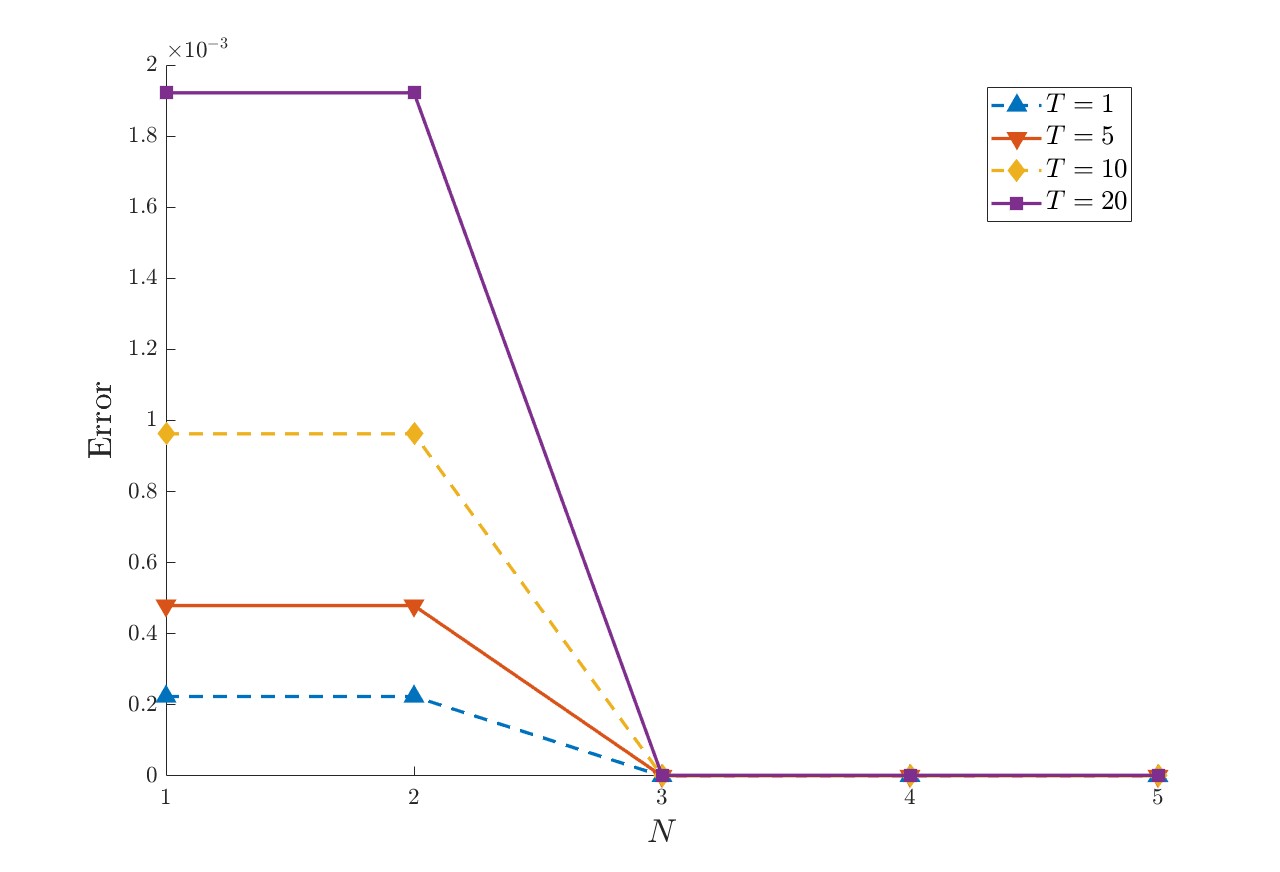}
		\caption{Left: The Carleman linearization error for the FPU chain model in \cref{eq: fpu}  in the time interval $[0,10]$  for different trunction level $N \leq 5$  with different level of nonlinearity, indicated by $\norm{F_3}$, which is controlled by changing $\alpha$ in \cref{eq: fpu}; Right: The Carleman linearization error for the FPU chain model in \cref{eq: fpu} for different truncation levels $N$ and different time duration $T$. Here we set $\norm{F_{3}}_{1}=1$.  }
		\label{FPUF3}	
\end{figure}

In the numerical tests, we initialize the FPU chain by a sinusoidal function,
\begin{align}
u\left(x,0\right) =  \frac{1}{10}\sin\left(\frac{2\pi x}{p+1}\right), \quad 
\frac{\partial}{\partial t}{u}\left(x,0\right) =  0.
\label{fpu_initial}
 \end{align}
We choose $k=1$ and set the length of the chain to  $p=7$, and test the Carleman linearization up to level $N=5$. The snapshots of the exact solution are displayed in \cref{FPU_Mesh}. Unlike the solution to the Burgers' equation depicted in \cref{BUR_mesh}, the solution to the FPU equation under these settings exhibits persistent oscillations.

To assess the impact of nonlinearity on error across various simulation durations, we tune the parameter $\alpha$ in \cref{eq: fpu}. The error from the Carleman linearization for each choice of $\alpha$  is shown in \cref{FPUF3}. We observed that even in the absence of dissipative conditions \eqref{eq: dissip}, the Carleman linearization converges when $\norm{F_{3}}$ remains sufficiently small, and divergence from the exact solution is observed for stronger nonlinearity. With numerical calculations, we find that this FPU chain is non-resonant, but with a small resonance parameter $\Delta= 2.3444\times 10^{-4}.$ Therefore it is quite remarkable that the Carleman linearization still shows convergence for modest nonlinearity.

We also studied the influence of total time duration $T$ by keeping $\norm{F_{3}}_{1}=1$. The numerical results, as shown in the right panel in \cref{FPUF3}, suggest that the error over a longer time interval often grows with $T$ (with a few exceptions), but the growth does not seem to be exponential.

\section{Summary and discussions}

In this paper, we have identified a new regime for dynamical systems where Carleman linearization can achieve linear convergence with the truncation level \(N\). This discovery can be utilized to establish an efficient quantum algorithm. The key to this convergence is a resonance parameter \(\Delta\), which, in the error bound, functions similarly to the dissipative parameter in previous analyses of the dissipative regime. However, these two regimes are not mutually exclusive; both dissipation and dispersion can be at play simultaneously. Additionally, our numerical results suggest that the current error bounds are likely pessimistic. For instance, the sparsity of \(F_2\) has not been taken into account.
There are other linear embedding schemes \cite{engel2021linear, giannakis2022embedding} for reducing nonlinear dynamical systems to linear ODEs. Extending the current analysis to study the truncation error of those schemes is another interesting direction to explore.

We have kept the focus of this paper mainly on the analysis of the error from the Carleman linearization. The actual implementation of the algorithm should also involve the preparation of $\bm y(0)$, a quantum algorithm for linear ODEs, fast-forwarding schemes, and efficient algorithms for the measurements of quantities of interest. We refer readers to \cite{krovi2022improved, liu2023efficient, an2022theoryof,liu2021efficient} and the references therein for the discussions of these aspects. It is also important to point out that the present approach and analysis are not limited to ODEs. Rather, such ODEs can arise from the semi-discrete approximations of many PDEs. The Burgers and KdV equations are already two important examples of this kind.

Our analysis primarily focuses on the dynamical systems described in \cref{eq:1.1}, where there exists a stable equilibrium at $\bm x=0$. Additionally, we can introduce an external force $F_0$, as discussed in \cite{liu2021efficient}. When $F_0$ is sufficiently small, it merely shifts the equilibrium point, allowing the same analytical approach to remain applicable. However, the situation becomes more complex when the system is exploring multiple equilibrium points, necessitating a different analytical framework.
Moreover, even in scenarios involving a single equilibrium, the occurrence of resonance can induce chaotic behavior, as illustrated in examples from \cite[Chap. 7]{ott2002chaos}. In such instances, any approximation method must contend with rapidly growing errors. Considering that the optimal complexity of a quantum algorithm for solving a linear ordinary differential equation (ODE) system, particularly via the Carleman linearization method, increases linearly with time $T$, it appears challenging for Carleman linearization to achieve fast or even any convergence in these contexts. This perspective is consistent with the insights presented in \cite{lewis2023limitations}.



\appendix
\section{Proof of Lemma \ref{Lem: Main lemma}}\label{a-proof}
\begin{proof}
We first suppose $A_{i,i}$ and $A_{j,j}$ do not share the same eigenvalues. We can assume $N$ is an even number, so $\frac{N}{2}=\floor{\frac{N}{2}}$.  The case when $N$ is odd can be proved similarly. Denote $s = \frac{N}{2}$. For the case $1\le k\le s$, we rewrite \cref{w_N-k} into

\begin{equation}\label{eq: w_N-k 2}
    \begin{aligned}
        &\bm w_{j-k} = \sum_{m_1=1}^{j-k}\mathcal{L}(G_{p_{m_1}}\bxi_{m_1}^{(k)})\\
        & + \chi_{k\ge2} \sum_{m_1=1}^{j-k-1}\left(\sum_{m_2=m_1+k}^{j-1}\mathcal{L}(G_{p_{m_1}}\bxi_{m_1}^{(k-1)},G_{p_{m_2}}\bxi_{m_2}^{(1)})+\sum_{m_2=m_1+2}^{j-k+1}\mathcal{L}(G_{p_{m_1}}\bxi_{m_1}^{(1)},G_{p_{m_2}}\bxi_{m_2}^{(k-1)})\right)\\
        & + \chi_{k\ge3} \sum_{m_1=1}^{j-k-1}\!\! \left( \sum_{m_2=m_1+k-1}^{j-2}\mathcal{L}(G_{p_{m_1}}\bxi_{m_1}^{(k-2)},G_{p_{m_2}}\bxi_{m_2}^{(2)})+\!\!\sum_{m_2=m_1+3}^{j-k+2}\!\!\mathcal{L}(G_{p_{m_1}}\bxi_{m_1}^{(2)},G_{p_{m_2}}\bxi_{m_2}^{(k-2)})  \right)\\
        & + \chi_{k\ge3} \sum_{m_1=1}^{j-k-2}\Bigg(\sum_{m_2=m_1+k-1}^{j-3}\sum_{m_3=m_2+2}^{j-1}\mathcal{L}(G_{p_{m_1}}\bxi_{m_1}^{(k-2)},G_{p_{m_2}}\bxi_{m_2}^{(1)},G_{p_{m_3}}\bxi_{m_3}^{(1)})\\
        & \qquad\qquad +\sum_{m_2=m_1+2}^{j-k}\sum_{m_3=m_2+k-1}^{j-1}\mathcal{L}(G_{p_{m_1}}\bxi_{m_1}^{(1)},G_{p_{m_2}}\bxi_{m_2}^{(k-2)},G_{p_{m_3}}\bxi_{m_3}^{(1)})\\
        & \qquad\qquad +\sum_{m_2=m_1+2}^{j-k}\sum_{m_3=m_2+2}^{j-k+2}\mathcal{L}(G_{p_{m_1}}\bxi_{m_1}^{(1)},G_{p_{m_2}}\bxi_{m_2}^{(1)},G_{p_{m_3}}\bxi_{m_3}^{(k-2)})\Bigg)\\
        & + \cdots 
        \quad + \chi_{k= s}\sum_{m_1=1}^{j-k-s+1}\sum_{m_2=m_1+2}^{j-k-s+3}\cdots \sum_{m_{k}=j-1}^{j-1}\mathcal{L}(G_{p_{m_1}}\bxi_{m_1}^{(1)},\cdots,G_{p_{m_{k}}}\bxi_{m_{k}}^{(1)}).
    \end{aligned}
\end{equation}
The indicator function simply indicates which terms will appear depending on the value of $k$.

The proof is by induction. For the base cases  $\bm w_{N-1}$ and $\bm w_{N-2}$, we know they are in the form from previous calculations, in particular \cref{eq: w_j-1,w_j-2 Fredholm}. Assume that the formula holds true for all $k=1,\cdots,c$, where $c<s-1$. We proceed by picking out the next equation in \cref{eq: evector} by back substitution, the equation $ A_{j-c-1,j-c-1}\bm w_{j-c-1}+A_{j-c-1,j-c}\bm w_{j-c}=\l \bm w_{j-c-1}$ implies that
\begin{equation*}
        (\l I-A_{j-c-1,j-c-1})\bm w_{j-c-1}=A_{j-c-1,j-c}\bm w_{j-c}
        = \left(\sum_{k=0}^{j-c-2}I^{\otimes j-c-2-k}\otimes F_2\otimes I^{\otimes k}\right)\bm w_{j-c}.
\end{equation*}
Now we try to find the form of $A_{j-c-1,j-c}\bm w_{j-c}$. For any $r\in\{1,\cdots,N-c-1\}$. Fix $j=N$, i.e., we work with the last block in \cref{eq: evector}, and we have
\begin{equation}\label{eq:xi^(c+1)}
    \begin{aligned}
         \bm\xi_{m_r}^{(c+1)} &= F_2\left(G_{p_{m_r}}\bxi_{m_r}^{(c)}\otimes \be_{i_{m_r+c+1}}+ \be_{i_{m_r}}\otimes G_{p_{m_r+1}}\bxi_{m_r+1}^{(c)}\right)\\
    &\quad +F_2\left(G_{p_{m_r}}\bxi_{m_r}^{(c-1)}\otimes G_{p_{m_{r+c}}}\xi_{m_{r+c}}\bxi^{(1)}+G_{p_{m_r}}\bxi_{m_r}^{(1)}\otimes G_{p_{m_{r+2}}}\bxi_{m_{r+2}}^{(c-1)} \right)\\
    &\quad +F_2\left(G_{p_{m_r}}\bxi_{m_r}^{(c-2)}\otimes G_{p_{m_{r+c-1}}}\bxi_{m_{r+c-1}}^{(2)}+G_{p_{m_r}}\bxi_{m_r}^{(2)}\otimes G_{p_{m_{r+3}}}\bxi_{m_{r+3}}^{(c-2)} \right)\\
    &\quad +\cdots
     = \sum_{a=0}^c F_2\left(G_{p_{m_r}}\bxi_{m_r}^{(a)}\otimes G_{p_{m_r+a+1}}\bxi_{m_r+a+1}^{(c-a)}\right).
    \end{aligned}
\end{equation}
As a result, we find that $A_{j-c-1,j-c}\bm w_{j-c}$ has a term in the following form $\sum_{k=1}^{j-c-1}\mathcal{L}\left(\bxi_{k}^{(c+1)}\right).$

\cref{eq:xi^(c+1)} presented the case where two consecutive vectors at position $m_r$ with orders adding up to $c$. Now we may look at the case where the orders sum up to $c-1$. Similarly,
we get $\bxi_{m_r}^{(c)}$. From the induction assumption, there exists another vector $G_{m_r'}\bxi_{m_r'}^{(1)}$. Therefore, we have $\mathcal{L}(\bxi_{m_r}^{(c)},G_{m_r'}\bxi_{m_r'}^{(1)})$. We pair this case together with the case where at position $m_r$ the term has order $c$ but $F_2$ applies to the vector that has order $1$. From these observation,  we find that $A_{j-c-1,j-c}\bm w_{j-c}$ has a term 
\[ \sum_{r=1}^{j-c-2}\sum_{r'=r+c+1}^{N-1} \left(\mathcal{L}(\bxi_{m_r}^{(c)},G_{p_{m_r'}}\bxi_{m_r'}^{(1)})+ \mathcal{L}(G_{p_{m_r}}\bxi_{m_r}^{(c)},\bxi_{m_r'}^{(1)})\right).\]

We repeat this argument and get the remaining terms in $A_{j-c-1,j-c}\bm w_{j-c}$. Note that 
\[ \bm w_{j-c-1} = (\l I-A_{j-c-1,j-c-1})^{-1}A_{j-c-1,j-c}\bm w_{j-c} \]
The invertibility of $ (\l I-A_{j-c-1,j-c-1})$ follows from the assumption of $\sigma(A_{j,j})\cap \sigma(A_{j-k,j-k})=\phi$ for all $k$. Next we look at $(\l I-A_{j-c-1,j-c-1})^{-1}\sum_{r=1}^{l-c-1}\mathcal{L}\left(\bxi_r^{c+1}\right)$. By spectrum decomposition of matrix $(\l I-A_{j-c-1,j-c-1})^{-1}$, we have
\begin{equation*}
    \begin{aligned}
        (\l I-A_{j-c,j-c})^{-1}&(\be_{i_1}\otimes\cdots\otimes \be_{i_{r-1}}\otimes \bxi_r^{(c+1)}\otimes\cdots\otimes \be_{i_j})\\
        &= \be_{i_1}\otimes\cdots\otimes \be_{i_{r-1}}\otimes G_{r:(r+c+1)}\bxi_r^{(c+1)}\otimes\cdots\otimes \be_{i_j} 
        = \mathcal{L}(G_{p_r}\bxi_r^{(c+1)}).
    \end{aligned}
\end{equation*}
Therefore,
\( (\l I-A_{j-c-1,j-c-1})^{-1}\sum_{r=1}^{N-c-1}\mathcal{L}\left(\bxi_r^{(c+1)}\right) = \sum_{r=1}^{N-c-1}\mathcal{L}\left(G_{p_r}\bxi_{r}^{(c+1)}\right) \).
Now we prove that for every $1\le a< c+1$, we have
\[  (\l I-A_{j-c-1,j-c-1})^{-1} \left(\oL(G_{p_r}\bxi_r^{(c+1-a)},\bxi_{q}^{(a)})+\oL(\bxi_r^{(c+1-a)},G_{p_{q}}\bxi_{q}^{(a)})\right)= \oL(G_{p_r}\bxi_r^{(c+1-a)},G_{p_{q}}\bxi_q^{(a)}). \]

A direct calculation shows that, 
\begin{align*}
         (\l I\!-\! A_{j-c-1,j-c-1})^{-1}\oL(G_{p_r}\bxi_r^{(c+1-a)},\bxi_{q}^{(a)})\!
          &=\!\sum_{\l_{i_k'}}\frac{1}{\l-\!(\l_{i_1'}+\cdots+\l_{i_{j-c-1}'})}\bigotimes_{k=1}^{j-c-1}{\be_{i_k'}}\left(\bigotimes_{k=1}^{j-c-1}{\bm f_{i_k'}}\right)^T\\
          & \be_{i_1}\otimes\cdots\otimes \be_{i_{r-1}}\otimes  G_{p_r}\bxi_r^{(c+1-a)}\otimes\cdots\otimes \bxi_{q}^{(a)}\otimes\cdots\otimes \be_{i_j}. 
\end{align*}
In order to make it nonzero, we must require that
\[i_1=i_1',\cdots, i_{r-1}=i_{r-1}', i_{r+c-a+2}=i_{r+1}',\cdots, i_{q-1} = i_{q-1}', \]
and
\(i_{q+a+1}=i_{q+1}',\cdots,i_{j}=i_{j-c-1}'. \)
Then the above expression is equivalent to
\begin{equation*}
    \begin{aligned}
        &\sum_{\l_{i_r'},\l_{i_q'}}\frac{1}{(\l_{i_{r}}+\cdots+\l_{i_{r+c+1-a}}+\l_{i_{q}}+\cdots+\l_{i_{q+a}})-(\l_{i_{r'}}+\l_{i_{q'}})}\\ 
        &\qquad\qquad\qquad\qquad\qquad\qquad\qquad (\be_{i_{r'}}\otimes \be_{i_{q'}})(\bbf_{i_{r'}}\otimes \bbf_{i_{q'}})^T G_{p_r}\bxi_{r}^{(c+1-a)}\otimes \bxi_q^{(a)}\\
        &= \sum_{\l_{i_r'},\l_{i_q'}}\frac{1}{(\l_{i_{r}}+\cdots+\l_{i_{r+c+1-a}}+\l_{i_{q}}+\cdots+\l_{i_{q+a}})-(\l_{i_{r'}}+\l_{i_{q'}})}(\be_{i_{r'}}\otimes \be_{i_{q'}})\bbf_{i_{r'}}^T \\ 
        &\qquad\qquad\qquad\qquad\qquad\qquad\qquad\frac{1}{(\l_{i_{r}}+\cdots+\l_{i_{r+c+1-a}})-\l_{i_r'}}\bxi_{r}^{(c+1-a)}\otimes \bbf_{i_{q'}}^T \bm \xi_q^{(a)}.
    \end{aligned}
\end{equation*}

Similarly, we simplify $(\l I- A_{j-c-1,j-c-1})^{-1}\left(\oL(\bxi_r^{(c-a)},G_{p_q}\bxi_{q}^{(a)})\right)$ to,
\begin{equation*}
    \begin{aligned}
        & \sum_{\l_{i_r'},\l_{i_q'}}\frac{1}{(\l_{i_{r}}+\cdots+\l_{i_{r+c+1-a}}+\l_{i_{q}}+\cdots+\l_{i_{q+a}})-(\l_{i_{r'}}+\l_{i_{q'}})}\\
        &\qquad\qquad\qquad\qquad\qquad(
        \be_{i_{r'}}\otimes \be_{i_{q'}})(\bbf_{i_{r'}}\otimes \bbf_{i_{q'}})^T \bxi_{r}^{(c+1-a)}\otimes G_{p_q}\bxi_q^{(a)}\\
        &= \sum_{\l_{i_r'},\l_{i_q'}}\frac{1}{(\l_{i_{r}}+\cdots+\l_{i_{r+c+1-a}}+\l_{i_{q}}+\cdots+\l_{i_{q+a}})-(\l_{i_{r'}}+\l_{i_{q'}})}\\
        &\qquad\qquad\qquad(\be_{i_{r'}}\otimes \be_{i_{q'}})\bbf_{i_{r'}}^T \bxi_{r}^{(c+1-a)}\otimes\frac{1}{(\l_{i_q}+\cdots+\l_{i_{q+a}})-\l_{i_{q'}}} \bbf_{i_{q'}}^T \bxi_q^{(a)}.
    \end{aligned}
\end{equation*}

By collecting these terms, we arrive at, 
\begin{equation*}
    \begin{aligned}
         &(\l I- A_{j-c-1,j-c-1})^{-1}\left(\oL(G_{p_r}\bxi_r^{(c-a)},\bxi_{q}^{(a)})+\oL(G_{p_r}\bxi_r^{(c+1-a)},\bxi_{q}^{(a)})\right)\\
         &= \sum_{\l_r'}\frac{1}{(\l_{i_{r}}+\cdots+\l_{i_{r+c+1-a}})-\l_{i_r'}} (\be_{i_r'}\otimes \be_{i_q'})(\bbf_{i_r'}\otimes \bbf_{i_q'})^T \bxi_r^{(c+1-a)}\otimes\bxi_q^{(a)}+\\
         &\qquad\qquad\qquad\qquad \sum_{\l_q'}\frac{1}{(\l_{i_q}+\cdots+\l_{i_{q+a}})-\l_{i_{q'}}} (\be_{i_r'}\otimes \be_{i_q'})(\bbf_{i_r'}\otimes \bbf_{i_q'})^T \bxi_r^{(c+1-a)}\otimes\bxi_q^{(a)}\\
         &= \oL(G_{p_r}\bxi_r^{(c+1-a)},G_{p_q}\bxi_q^{(a)}).
    \end{aligned}
\end{equation*}
Therefore, for $a=1,\cdots,\floor{\frac{c}{2}}$, we know $\bm w_{j-k}$ with $k=c+1$ contains
\begin{equation*}
    \begin{aligned}
        &\sum_{m_1=1}^{j-k-1}\left(\sum_{m_2=m_1+k}^{j-1}\mathcal{L}(G_{p_{m_1}}\bxi_{m_1}^{(k-1)},G_{p_{m_2}}\bxi_{m_2}^{(1)})+\sum_{m_2=m_1+2}^{j-k+1}\mathcal{L}(G_{p_{m_1}}\bxi_{m_1}^{(1)},G_{p_{m_2}}\bxi_{m_2}^{(k-1)})\right)\\
        &\quad+ \sum_{m_1=1}^{j-k-1} \left( \sum_{m_2=m_1+k-1}^{j-2}\mathcal{L}(G_{p_{m_1}}\bxi_{m_1}^{(k-2)},G_{p_{m_2}}\bxi_{m_2}^{(2)})+\sum_{m_2=m_1+3}^{j-k+2}\mathcal{L}(G_{p_{m_1}}\bxi_{m_1}^{(2)},G_{p_{m_2}}\bxi_{m_2}^{(k-2)})  \right) +  \\
        &\quad +\cdots +\sum_{m_1=1}^{j-k-1}\Big(\sum_{m_2=m_1+k-\floor{\frac{c}{2}}+1}^{j-\floor{\frac{c}{2}}}\mathcal{L}(G_{p_{m_1}}\bxi_{m_1}^{(k-\floor{\frac{c}{2}})},G_{p_{m_2}}\bxi_{m_2}^{(\floor{\frac{c}{2}})})
        \\
&\qquad\qquad\qquad\qquad\qquad\qquad\qquad\qquad
        +\sum_{m_2=m_1+\floor{\frac{c}{2}}+1}^{j-k+\floor{\frac{c}{2}}}\mathcal{L}(G_{p_{m_1}}\bxi_{m_1}^{(\floor{\frac{c}{2}})},G_{p_{m_2}}\bxi_{m_2}^{(k-\floor{\frac{c}{2}})})\Big).
    \end{aligned}
\end{equation*}

With similar arguments and using spectral decomposition and pairing terms, we can show that for any $k_1+\cdots+k_s=c+1$ and $k_1,\cdots,k_s\ge1$.
\begin{multline*}
    (\l I-A_{j-k-1,j-k-1})^{-1} \sum_{p(\bxi^{(k_1)})}\oL(\bxi_{m_1}^{(k_1)},G_{p_{m_2}}\bxi_{m_2}^{(k_2)},\cdots,G_{p_{m_s}}\bxi_{m_s}^{(k_s)}) \\ = \oL(G_{p_{m_1}}\bxi_{m_1}^{(k_1)},G_{p_{m_2}}\bxi_{m_2}^{(k_2)},\cdots,G_{p_{m_s}}\bxi_{m_s}^{(k_s)}).
\end{multline*}
Therefore, we see that  when $k=c+1$, $\bm w_{j-k}$ still satisfies \cref{eq: w_N-k 2}. By induction, the proof for the case $1\le k\le s$ is done. Now we look at the case $s\le k\le N-1$. 
\begin{equation*}
    \begin{aligned}
        \bm w_{s} &= \sum_{m_1=1}^{j-s}\mathcal{L}(G_{p_{m_1}}\bxi_{m_1}^{(s)})\\
        &\quad + \sum_{m_1=1}^{j-s-1}\left(\sum_{m_2=m_1+s}^{j-1}\mathcal{L}(G_{p_{m_1}}\bxi_{m_1}^{(s-1)},G_{p_{m_2}}\bxi_{m_2}^{(1)})+\sum_{m_2=m_1+2}^{j-s+1}\mathcal{L}(G_{p_{m_1}}\bxi_{m_1}^{(1)},G_{p_{m_2}}\bxi_{m_2}^{(s-1)})\right)\\
        &\quad + \sum_{m_1=1}^{j-s-1} \left( \sum_{m_2=m_1+s-1}^{j-2}\mathcal{L}(G_{p_{m_1}}\bxi_{m_1}^{(s-2)},G_{p_{m_2}}\bxi_{m_2}^{(2)})+\sum_{m_2=m_1+3}^{j-k+2}\mathcal{L}(G_{p_{m_1}}\bxi_{m_1}^{(2)},G_{p_{m_2}}\bxi_{m_2}^{(s-2)})  \right)\\
        &\quad+ \sum_{m_1=1}^{j-s-1}\left(\sum_{m_2=m_1+s-2}^{j-3}\mathcal{L}(G_{p_{m_1}}\bxi_{m_1}^{(s-3)},G_{p_{m_2}}\bxi_{m_2}^{(3)})+\sum_{m_2=m_1+4}^{j-s+3}\mathcal{L}(G_{p_{m_1}}\bxi_{m_1}^{(3)},G_{p_{m_2}}\bxi_{m_2}^{(s-3)})\right)\\
        &\quad + \cdots
         + \sum_{m_1=1}^{1}\sum_{m_2=m_1+2}^{3}\cdots \sum_{m_{k}=j-1}^{j-1}\mathcal{L}(G_{p_{m_1}}\bxi_{m_1}^{(1)},\cdots,G_{p_{m_{k}}}\bxi_{m_{k}}^{(1)}).
    \end{aligned}
\end{equation*}

From
\( \bm w_{j-k-1} = (\l I-A_{j-k-1,j-k-1})^{-1}A_{j-k-1,j-k}\bm w_{j-k}, \)
we find that the order a of $\bxi^{(a)}$ still sums up to $k$ for each $\bm w_{j-k}$. Therefore, \cref{w_N-k} still holds. 

Thus, the proof is completed under the assumption that $\sigma(A_{i,i})\cap \sigma(A_{j,j})=\phi$ for all $i<k$. Now we show that the assumption can be removed by using Fredholm alternative. As can be seen from \eqref{G12},  $(\l-A_{j-1,j-1})$ is always invertible under the nonresonance condition. Equation \eqref{w_j-2 Fredholm} shows that Fredholm alternative works for $k=2$. Assume that $\l I - A_{j-k,j-k}$ is psuedo-invertible for $k=1,\cdots,c$ (i.e., fulfilling the Fredholm alternative).
Consider the next case where $k=c+1$: 
\[ (\l I - A_{j-k,j-k})\bm w_{j-k}=A_{j-k,j-k+1}\bm w_{j-k+1}, \]
where $(\l I-A_{j-k,j-k})$ has a nontrivial null space. Denote $\bm h = \left(\bbf_{i_1'}\otimes\cdots\otimes \bbf_{i_{j-k}'}\right)^T$,  a left eigenvector of $\l I-A_{j-k,j-k}$. Since $\bm w_{j-k+1}$ follows the form of  \cref{eq: w_N-k 2}, $A_{j-c-1,j-c}\bm w_{j-c}$ follows the pattern, 
\[ \sum_{k=1}^{j-c-1}\mathcal{L}\left(\bxi_{k}^{(c+1)}\right), \quad \sum_{r=1}^{j-c-2}\sum_{r'=r+c+1}^{N-1} \left(\mathcal{L}(\bxi_{m_r}^{(c)},G_{p_{m_r'}}\bxi_{m_r'}^{(1)})+ \mathcal{L}(G_{p_{m_r}}\bxi_{m_r}^{(c)},\bxi_{m_r'}^{(1)})\right),\quad \cdots\]

Let us first look at the term $\oL\left(\bxi_1^{(c+1)}\right)=\bxi_1^{(c+1)}\otimes \be_{i_{c+3}}\otimes\cdots\be_{i_j}$. If $\bm h\cdot \oL\left(\bxi_1^{(c+1)}\right)\ne0$, then
$i_2'=i_{c+3}$, $i_3'=i_{c+4}$,..., $i_{j-c-1}'=i_j$, which leads to $\l_{i_2'}=\l_{i_{c+3}}$, $\l_{i_3'}=\l_{i_{c+4}}$,..., $\l_{i_{j-c-1}'}=\l_{i_j}$. Then
\[ \l=\l_{i_1}+\cdots+\l_{i_j}=\l_{i_1'}+\cdots+\l_{i_{j-c-1}'}\To \l_i'=\l_{i_1}+\cdots+\l_{i_{c+2}} \]
contradicting the nonresonance condition. Similary, we conclude that $\bm h\cdot \sum_{k=1}^{j-c-1}\oL\left(\bxi_k^{(c+1)}\right)=0$. 

Next we look at the term $\mathcal{L}(\bxi_{m_r}^{(c)},G_{p_{m_r'}}\bxi_{m_r'}^{(1)})+ \mathcal{L}(G_{p_{m_r}}\bxi_{m_r}^{(c)},\bxi_{m_r'}^{(1)})$. Using the same argument as in  \cref{w_j-2 Fredholm}, we know $\bm h\cdot \mathcal{L}(G_{p_{m_r}}\bxi_{m_r}^{(c)},\bxi_{m_r'}^{(1)}) = 0$. Now we claim that
\begin{multline}\label{Fredholm c-1}
    \bm h \cdot \Big(\oL(G_{1:c}\bxi_1^{(c-1)},G_{c+1,c+2}\bxi_{c+1}^{(1)}, \bxi_{c+3}^{(1)})+\oL(\bxi_1^{(c-1)},G_{c+1,c+2}\bxi_{c+1}^{(1)}, G_{c+3,c+4}\bxi_{c+3}^{(1)})\\ 
    +\oL(G_{1:c}\bxi_1^{(c-1)},\bxi_{c+1}^{(1)}, G_{c+3,c+4}\bxi_{c+3}^{(1)})\Big)=0.
\end{multline} 
Notice  that 
\begin{gather*}
    \oL(G_{1:c}\bxi_1^{(c-1)},G_{c+1,c+2}\bxi_{c+1}^{(1)}, \bxi_{c+3}^{(1)}) = G_{1:c}\bxi_1^{(c-1)}\otimes G_{c+1,c+2}\bxi_{c+1}^{(1)}\otimes\bxi_{c+3}^{(1)}\otimes\cdots\otimes \be_{i_j},\\
    \oL(\bxi_1^{(c-1)},G_{c+1,c+2}\bxi_{c+1}^{(1)}, G_{c+3,c+4}\bxi_{c+3}^{(1)}) = \bxi_1^{(c-1)}\otimes G_{c+1,c+2}\bxi_{c+1}^{(1)}\otimes G_{c+3,c+4}\bxi_{c+3}^{(1)}\otimes\cdots\otimes\be_{i_j},\\
    \oL(G_{1:c}\bxi_1^{(c-1)},\bxi_{c+1}^{(1)}, G_{c+3,c+4}\bxi_{c+3}^{(1)}) =  G_{1:c}\bxi_1^{(c-1)}\otimes \bxi_{c+1}^{(1)}\otimes G_{c+3,c+4}\bxi_{c+3}^{(1)}\otimes\cdots\otimes \be_{i_j}.
\end{gather*}
Thus we require $i_4'=i_{c+4}$, ... $i_{j-c-1}'=i_j$ for the inner product not be zero.
Then 
\begin{equation*}
    \begin{aligned}
        &\bm h\cdot \oL(G_{1:c}\bxi_1^{(c-1)}, G_{c+1,c+2}\bxi_{c+1}^{(1)}, \bxi_{c+3}^{(1)})
        =\bm h\cdot G_{1:c}\bxi_1^{(c-1)}\otimes G_{c+1,c+2}\bxi_{c+1}^{(1)}\otimes\bxi_{c+3}^{(1)}\otimes\cdots\otimes \be_{i_j}\\
        &\qquad =\frac{1}{\l_{i_1}+\cdots+\l_{i_c}-\l_{i_1'}}\bbf_{i_1'}^T\bxi_{1}^{(c-1)}\cdot\frac{1}{\l_{i_{c+1}}+\l_{i_{c+2}}-\l_{i_2'}}\bbf_{i_2'}^T\bxi_{c+1}^{(1)}\cdot \bbf_{i_3'}^T\bxi_{c+3}^{(1)}
    \end{aligned}
\end{equation*}
Similarly, we obtain,
\begin{align*}
    &\bm h\cdot \oL(\bxi_1^{(c-1)},G_{c+1,c+2}\bxi_{c+1}^{(1)}, G_{c+3,c+4}\bxi_{c+3}^{(1)}) \\
    &=\bbf_{i_1'}^T\bxi_{1}^{(c-1)}\cdot\frac{1}{\l_{i_{c+1}}+\l_{i_{c+2}}-\l_{i_2'}}\bbf_{i_2'}^T\bxi_{c+1}^{(1)}\cdot \frac{1}{\l_{i_{c+3}}+\l_{i_{c+4}}-\l_{i_3'}}\bbf_{i_3'}^T\bxi_{c+3}^{(1)},\\
    & \bm h\cdot \oL(G_{1:c}\bxi_1^{(c-1)},\bxi_{c+1}^{(1)}, G_{c+3,c+4}\bxi_{c+3}^{(1)}) \\
    &=\frac{1}{\l_{i_1}+\cdots+\l_{i_c}-\l_{i_1'}}\bbf_{i_1'}^T\bxi_{1}^{(c-1)}\cdot\bbf_{i_2'}^T\bxi_{c+1}^{(1)}\cdot \frac{1}{\l_{i_{c+3}}+\l_{i_{c+4}}-\l_{i_3'}}\bbf_{i_3'}^T\bxi_{c+3}^{(1)}.
\end{align*}

As a result, the left hand side of  \cref{Fredholm c-1} can be simplified to
\begin{equation}
    \begin{aligned}
        \frac{(\l_{i_1}+\cdots+\l_{i_{c+4}})-(\l_{i_1'}+\l_{i_2'}+\l_{i_3'})}{(\l_{i_1}+\cdots+\l_{i_c}-\l_{i_1'})(\l_{i_{c+1}}+\l_{i_{c+2}}-\l_{i_2'})(\l_{i_{c+3}}+\l_{i_{c+4}}-\l_{i_3'})}\bbf_{i_1'}^T\bxi_{1}^{(c-1)}\bbf_{i_2'}^T\bxi_{c+1}^{(1)}\bbf_{i_3'}^T\bxi_{c+3}^{(1)}=0.
    \end{aligned}
\end{equation}
One may prove for the general case using similar pairing techniques. Thus, psuedo-inverse of $\l I-A_{j-k,j-k}$ always exists for any $k=1,\cdots,j-1$.

\end{proof}

\section{Proof of Lemma \ref{lem: wj-k}} \label{b-proof}
Notice that right hand side of \cref{wj-k-cartalan} is 
\[ \left( \frac{\norm{F_2}}{\Delta} \right)^{k} \sum_{\overset{k_1+k_2+\cdots+k_r=k}{r= j-k}}  C(k_1) C(k_2) \cdots C(k_r), \]
where $C(k_i)$ is the Catalan numbers for $i=1,\cdots,r$. This product can be conveniently recovered from  the generating function for Catalan numbers, given by \cite{xie2021new},
\[ G(x)=\sum_{n=0}^\oo \mat{2n\\ n}\frac{x^n}{n+1}=\frac{1-\sqrt{1-4x}}{2x}. \]
Then the coefficient of $x^k$ in $\left(\frac{1-\sqrt{1-4x}}{2x}\right)^r$ coincides with the sum of the product, 
giving the explicit formula,
\[ \sum_{\overset{k_1+k_2+\cdots+k_r=k}{r= j-k}}  C(k_1) C(k_2) \cdots C(k_r) = \frac{r}{2k+r}\mat{2k+r\\ k}=\frac{j-k}{j+k}\mat{j+k\\ k}. \]

\end{document}